\documentclass[a4paper,11pt]{article}

\usepackage[utf8]{inputenc}
\usepackage[T1]{fontenc}
\usepackage{natbib}
\usepackage[english]{babel}
\usepackage{setspace}

\usepackage{geometry} 
\usepackage{rotating}

\usepackage{amsmath,amssymb,amsthm}
\usepackage{enumerate}
\usepackage{enumitem}
\usepackage{graphicx}
\usepackage{tcolorbox}
\usepackage{textcomp}
\usepackage{wasysym}
\usepackage{hyperref}
\usepackage[ruled, vlined]{algorithm2e}
\usepackage{multicol}
\usepackage{xcolor}
\usepackage{tabulary}
\usepackage{colortbl}
\usepackage{booktabs}
\usepackage{dsfont}

\usepackage{tikz}
\usetikzlibrary{calc,shapes,backgrounds,arrows,automata,shadows,positioning}

\definecolor{redorg}{RGB}{215,48,39}
\definecolor{orangeorg}{RGB}{253,174,97}
\definecolor{blueind}{RGB}{69,117,233}
\definecolor{cyanind}{RGB}{116,173,209}
\definecolor{greenind}{RGB}{171,217,233}

\newcommand{\pr}[1]{\mathbb{P}(#1)}
\newcommand{\ZR}[1]{Z^{I}_{#1}}
\newcommand{\ZL}[1]{Z^{O}_{#1}}
\newcommand{\XR}[1]{X^{I}_{#1}}
\newcommand{\XL}[1]{X^{O}_{#1}}

\newcommand{\pir}[1]{\pi^I_{#1}}
\newcommand{\pil}[1]{\pi^O_{#1}}
\newcommand{\alphar}[1]{\alpha^I_{#1}}
\newcommand{\alphal}[1]{\alpha^O_{#1}}
\newcommand{\taur}[1]{\tau^I_{#1}}
\newcommand{\taul}[1]{\tau^O_{#1}}
\newcommand{\nbr}{n_I}
\newcommand{\nbl}{n_O}
\newcommand{\QR}{Q_I}
\newcommand{\QL}{Q_O}
\newcommand{\obs}{\mathbf{X}}
\newcommand{\lat}{\mathbf{Z}}
\newcommand{\aff}{A}

\newcommand{\KL}{\mathrm{KL}}
\newcommand{\SBM}{\operatorname{SBM}}
\newcommand{\ICL}{\operatorname{ICL}}
\newcommand{\AUC}{\operatorname{AUC}}

\newcommand{\Ind}{\mathrm{Ind}}
\newcommand{\MLVSBM}{\mathrm{MLVSBM}}
\newcommand{\pen}{\operatorname{pen}}
\newcommand{\ind}{\mathds{1}}
 \DeclareMathOperator*{\argmax}{arg\,max}

\newcommand{\OA}[2]{\textcolor{gray}{#1}\textcolor{black}{#2}}

\newtheorem{proposition}{Proposition}
\newtheorem*{remark}{Remark}

\title{A Stochastic Block Model Approach for the Analysis of Multilevel Networks: an Application to the Sociology of Organizations}
\author{Saint-Clair Chabert-Liddell\thanks{Université Paris-Saclay, AgroParisTech, INRAE, UMR MIA-Paris, 75005, Paris, France} \and Pierre Barbillon\footnotemark[1] \and Sophie Donnet\footnotemark[1] \and Emmanuel Lazega\thanks{Institut d'Études Politiques de Paris, France}\\ 
}
\date{}

\begin{document}

\maketitle

\begin{abstract}
 \OA{}{ A multilevel network is defined as the junction of two interaction networks, one level representing the interactions between individuals and the other the interactions between organizations. The levels are linked by an affiliation relationship, each individual belonging to a unique organization. 
A new Stochastic Block Model is proposed as a unified probalistic framework tailored for multilevel networks.
This model contains latent blocks accounting for heterogeneity in the patterns of connection within each level and introducing dependencies between the levels. The sought connection patterns are not specified a priori which makes this approach flexible.
Variational methods are used for the model inference and an Integrated Classified Likelihood criterion  is developed for choosing the number of blocks and also for deciding whether the two levels are dependent or not.
A comprehensive simulation study exhibits the benefit of considering this approach, illustrates the robustness of the clustering and highlights the reliability of the criterion used for model selection.
This approach is applied on a sociological dataset collected during a television program trade fair, the inter-organizational level being the economic network between companies and the inter-individual level being the informal network between their representatives. It brings a synthetic representation of the two networks unraveling their intertwined structure and confirms the \emph{coopetition} at stake. }  
\end{abstract}

\emph{Keywords: }Latent variable model, Hierarchical modeling, Social network, Variational inference 

\section{Introduction}

The statistical analysis of network data has been a hot topic for the last decade.  The last few years witnessed a growing interest for multilayer networks   \citep[see][]{kivela2014multilayer, bianconi2018multilayer, giordano2019analyzing}. A particular case of multilayer networks are  multilevel networks \OA{}{where each level is a layer and an affiliation relationship represents the inter-layer}. 
\OA{}{Multilevel networks are used across many fields such as  sociology \citep{lazega2015multilevel} or environmental science \citep{hileman2018network}. In particular they arise in the sociology of organizations and collective action when willing to study jointly the social network  of individuals and the interaction network of organizations the individuals belong to. Indeed,  the individuals  not only interact with each others but  are also members of interacting  organizations. This approach is quite generic in the social sciences and all the phenomena of \emph{coopetition} and the maintenance of social inequalities can fall within the scope of this approach \citep{lazega2016structural}. It is also gaining attention as a way to articulate social network analysis and the life course studies \citep{vacchianoworking}.} Following \cite{lazega2015multilevel}, one might think that these two types of interactions (between individuals  and  between organizations)  are interdependent, the  individuals shaping their organizations and the organizations having an influence on the individuals. We aim to propose a statistical model for multilevel networks in order to understand how the two levels are intertwined and how one level impacts the other. 

In what follows,  a multilevel network is defined as the collection of an inter-individual network, an inter-organizational network and the affiliation of the individuals to the organizations. Besides, we assume that  the individuals belong to a unique organization.
Such a dataset is studied by \citet{lazega2008catching}, some researchers in cancerology being the individuals and their  laboratories the organizations. 
\citet{brailly2016embeddedness} deal with another dataset concerned with  the economic network of audiovisual firms and the informal network of their sales representatives during a trade fair. This latter dataset will be analyzed in this paper.

In the last years,  the Stochastic Block Model    \citep[SBM developed  by][]{holland1983stochastic,snijders1997estimation}  has become a popular tool to model  the heterogeneity of connection in a network, assuming that the actors at stake are divided into blocks (clusters) and that the members of a same  block  share a similar profile of connectivity. 
Compared to other graph clustering methods such as modularity maximization, hierarchical clustering or spectral clustering \citep[see][and references therein]{kolaczykstatistical},  the SBM is a generative model\OA{}{, it shares with the generalized blockmodeling \citep{doreian2005generalized} that they} can \OA{}{both} fit to a wide range of topologies since they gather  into blocks the nodes that are structurally equivalent.  \OA{}{However, contrary to the generalized blockmodeling which seeks a pre-specified structure in the network with given ideal blocks,
the SBM is agnostic and is aimed to unravel any kind of block structure which may shape the data.} This includes but is not restricted to the detection of assortative communities where the probability of connection within a block is higher than the probability of connection between the blocks. \OA{}{  Moreover, the probabilistic generative model allows the modeler to have a unified framework for model selection and natural extensions such as dealing with non binary dyads and link prediction.}  The SBMs have been extended to  particular types of multilayer networks :  \cite{barbillon2017stochastic} propose an SBM for multiplex networks and \cite{matias2017statistical} an SBM for time-evolving networks. In this paper, we propose an SBM suited to multilevel networks (MLVSBM).

\paragraph{Our contribution}
 In a few words, we model the heterogeneity in the  inter-individual and inter-organizational connections by introducing blocks of individuals and blocks of organizations, the blocks containing homogeneous groups of actors (individuals or organizations)  with respect to  their connectivity.  
The two levels are assumed to be interdependent through  their latent blocks. More specifically, the latent blocks of the inter-individual level depend on the latent blocks of the inter-organizational level and the affiliation.  This bi-clustering approach allows us to determine how groups of organizations influence the connectivity patterns of their individuals. \OA{}{Note that the hierarchical model does not assume a causal effect of the blocks of organizations on the blocks of individuals but an interdependence between the two sets of blocks.}
 
Due to the latent variables, the estimation of the parameters is a complex task. 
We resort to a variational version of the Expectation-Maximization (EM) algorithm.   For  the SBM, the variational approach \citep{jordan1999introduction,blei2017variational} has proven its efficiency for deriving maximum likelihood estimates  \citep{daudin2008mixture, mariadassou2010uncovering,  barbillon2017stochastic} and for Bayesian inference \citep{latouche2012variational, come2015model}.  
\OA{}{In the latent block model which is suited for bipartite network, the variational estimates have also been successfully applied \citep{govaert2008block}.}
In this paper, we obtain approximate maximum likelihood  estimates by an ad-hoc version of the  variational EM algorithm.

Another important task is the choice of the number of blocks. We propose  an adapted version of the Integrated Complete Likelihood (ICL) criterion. First developed by \cite{biernacki2000assessing} for mixture models as an alternative to the Bayesian Information Criterion (BIC), it was then adapted by \cite{daudin2008mixture}   to the SBM. The ICL has since illustrated its efficiency and relevance for various SBMs and their extensions such as multiplex network \citep{barbillon2017stochastic}, dynamic SBM \citep{matias2017statistical} or degree corrected SBM \citep{yan2016bayesian}. A further reference for dynamic SBMs is \cite{bartolucci2018dealing}.
\OA{}{Besides, a critical issue in sociology is to verify the multilevel interdependence hypothesis in a multilevel network, i.e. if the two levels (inter-individual and inter-organizational) should be analyzed jointly or if a separate analysis is sufficient. We thus propose a criterion to decide whether the two levels are independent or not.}

\paragraph{Related works}

The term \emph{multilevel network} arises in the statistical  literature for a wide variety of complex networks.  For instance,
 \citet{zijlstra2006multilevel} adapt the p2-model to handle multiple  observations of a network,  
 \cite{sweet2014hierarchical} extend  the Mixed Membership Stochastic Block Model \citep{airoldi2008mixed} to the hierarchical network model framework \citep{sweet2013hierarchical}  for the same type of data. 
\citet{snijders2017stochastic} discusses the use of the stochastic actor-oriented model \citep{snijders2001statistical} for temporal and multivariate networks.

When dealing with the  multilevel networks we defined before,   \citet{wang2013exponential} adopt an exponential random graph model (ERGM) strategy that is used in applications across many fields such as environmental science \citep{hileman2018network} or sociology \citep[chapter 10-11, 13-14]{lazega2015multilevel}.  When focusing on a clustering approach,  \cite{vziberna2014blockmodeling} develops three general approaches for blockmodeling multilevel networks. First, the separate analysis consists in clustering the levels separately or using the clustering of one level on the other. Second, the conversion approach converts the level of the organizations into a new kind of interaction between individuals, the interactions are then aggregated into a single layer network; this is close to the approach taken by  \cite{barbillon2017stochastic} who  transform the inter-organizational network into an inter-individual network thus adopting a multiplex network approach (the individuals interconnect directly or through the organizations they belong to).  
The third approach  is called the true multilevel approach and is the closest to the one we propose on this paper.  

\OA{}{\citet{vziberna2014blockmodeling} and the extensions in \citet{vziberna2019blockmodeling,vziberna2020k} to  a more general set of multilayer networks (called linked networks) 
use a generalized blockmodeling framework \citep{doreian2005generalized}.  
Contrary to this deterministic approach, we resort to a probalistic generative model for all the reasons stated above. The MLVSBM additionally provides us with a natural criterion for detecting the interdependence between the two levels. Furthermore, we
 explicitly take into account the constraint of having a unique affiliation per individual inherent to these multilevel datasets
and do not consider the affiliation as a bipartite network.}

Also, note that  the multiplex SBM approach applied to a multilevel network suggested by \cite{barbillon2017stochastic} is only applicable when the numbers of individuals and organizations are close. Indeed it requires to  duplicate the data of the inter-organizational level to fit the size of the inter-individual level. Furthermore, it only provides a clustering of the individuals and not two clusterings, one of the individuals and one of the organizations. \OA{}{In contrast, our MLVSBM  does not need to transform the data into a multiplex network and is able to obtain a clustering of the nodes within each level}.

 If we release the constraint of the unique affiliation, then the inter-level can be modeled by a latent block model and we obtain a particular case of the multipartite SBM of \cite{bar2018block}. However,  the interactions between individuals and organizations are considered at the same level as the affiliations, and the clustering might be strongly influenced  by the number of individuals   in each  organization.  

Finally, our work is also different from the SBM with edges covariates \citep{mariadassou2010uncovering} with the individuals as nodes and the inter-organizational network as edges covariates. Indeed, in that case, the clustering obtained for the individuals is the remaining structure of the inter-individual level once the effect of the covariates has been taken into account. In addition, this model does not provide a clustering of the organizations.

\paragraph{Outline of the paper}  The paper is organized as follows. The  SBM adapted to multilevel networks (MLVSBM) is defined in Section \ref{section:2}. We also give conditions  guaranteeing  the  independence between levels and  the  identifiability  of the parameters.
The inference strategy and the model selection criterion  are provided  in Section \ref{sec:inf}. 
The proof of the independence between levels, of the  identifiability and  the details on the variational EM and the ICL criterion are postponed to the Appendix sections. In Section \ref{section:simulation}, we present  an extensive simulation study illustrating the relevance of our inference method, model selection criterion and procedure. Section \ref{section:application} is dedicated  to the analysis of a sociological dataset  by our MLVSBM. 
Finally we discuss our contribution and future works in Section \ref{section:discussion}.

\section{A multilevel stochastic block model (MLVSBM)}\label{section:2}

\paragraph{Dataset}
\noindent Let us consider  $\nbr$ individuals involved in $\nbl$ organizations. 
We encode the networks into adjacency matrices as follows. Let $\XR{}$ be  the binary $\nbr \times \nbr$ matrix   representing the inter-individual network. $\XR{}$ is  such that :  $\forall(i, i') \in \{1, \dots, \nbr\}^2$: 
\begin{equation}\label{eqXR}
\XR{ii'} = \begin{cases}
  1 & \text{if there is an interaction from individual  $i$ to  individual $i'$}, \\
  0 & \text{otherwise}.	
 \end{cases} 
\end{equation}
$\XL{}$ is the  binary   $\nbl \times \nbl$  matrix representing the inter-organizational network,  $\forall(j, j') \in \{1, \dots, \nbl \}^2$: 
\begin{equation}\label{eqXO}
\XL{jj'} = \begin{cases}
  1 & \text{if there is an interaction from organization $j$ to organization $j'$,} \\
 0& \text{otherwise}.
 \end{cases}
\end{equation}

\begin{remark}
In general, no  self-loop are considered in the network, thus the interactions are defined for $i\neq i'$ and $j \neq j'$.  Moreover, if the interactions are undirected then
\begin{equation*}
    \XR{ii'} = \XR{i'i}  \quad \forall(i, i')    \quad \mbox{ or/and }  \quad 
    \XL{jj'} = \XL{j'j}  \quad \forall(j, j').
\end{equation*}
In what follows, we present the methodology for   undirected networks. However, all the results can be adapted to directed networks without any difficulty.  
\end{remark}

\noindent Let $\aff$ be the affiliation matrix. $\aff$ is a $\nbr \times \nbl$ matrix such that: 
\begin{equation*}
    \aff_{ij} = \begin{cases}
    1  & \text{if individual i belongs to organization j,} \\
    0  & \text{otherwise} \end{cases}.
\end{equation*}
$\aff$ is such that $\forall i = 1,\dots,\nbr$, $\sum_{j=1}^{\nbl}  A_{ij}=1$
since we assume that any individual belongs to a unique organization. A synthetic view of a generic dataset is provided in Figure \ref{fig:Matrice_Chercheur_Labo}.

\begin{figure}
    \centering

 {
 \begin{tabular}{|c||ccc||ccc|}
\multicolumn{1}{c }{}& \multicolumn{3}{c }{$\overbrace{\mbox{ \hspace{6em}}}^{\nbr}$}&  \multicolumn{3}{c }{$\overbrace{\mbox{ \hspace{6em}}}^{\nbl}$}\\

\hline

\cellcolor{blue!15} Individual $1$      &0 &              &  1   &0 &010   & 0\\
\cellcolor{blue!15} \vdots              &  & $\XR{ii'}$   &      &  &$A_{ij}$ &    \\
\cellcolor{blue!15}                     &  &              &      &  &         &    \\
\cellcolor{blue!15} Individual  $\nbr$  &1 &              &0  &  0 &      --      &01\\
\hline
\hline
\cellcolor{red!15}  Organization   $1$  & &                      &    &0  & & 1\\
\cellcolor{red!15}  \vdots              & &                        &   & & $\XL{jj'}$ &\\
\cellcolor{red!15}                      & &              &   &      &                &\\
\cellcolor{red!15} Organization $\nbl$  & &                      &    & 1 & & 0\\
\hline
\hline
& \cellcolor{blue!15}  \begin{turn}{-90}Individual $1$\end{turn}& \cellcolor{blue!15}  $\cdots$ & \cellcolor{blue!15}\begin{turn}{-90} Individual  ${\nbr}$ \end{turn}&  \cellcolor{red!15} \begin{turn}{-90}  Organization   $1$\end{turn}&$\cdots$ \cellcolor{red!15} & \cellcolor{red!15}  \begin{turn}{-90}  Organization $\nbl$\end{turn}\\
\hline

\end{tabular}
}

    \caption{Matrix representation of a multilevel network}
    \label{fig:Matrice_Chercheur_Labo}
\end{figure}

\vspace{1em}

We propose a joint modeling of the inter-individual and inter-organizational networks based on an extension of the SBM.
More precisely,  assume that the  $\nbl$  organizations  are divided into $\QL$ blocks and that the $\nbr$ individuals are  divided into $\QR$ blocks. Let $\ZL{} = (\ZL{1}, \dots, \ZL{\nbl})$  and  $\ZR{} = (\ZR{1}, \dots, \ZR{\nbr})$  be   such that    $\ZL{j} = l$  if organization $j$ belongs to block $l$  ($l \in \{1,\dots,\QL\}$) and $\ZR{i} = k$  if individual $i$ belongs to block $k$ ($k \in \{1,\dots,\QR\}$). 

Given these clusterings, we assume that the interactions between organizations and the interactions between individuals are independent and distributed as follows: 
\begin{equation}\label{eq:connec}
    \begin{array}{lcr}
    \mathbb{P}(\XL{jj'} = 1 | \ZL{j}, \ZL{j'}) &=& \alphal{\ZL{j}\ZL{j'}}, \\
    \mathbb{P}(\XR{ii'} = 1 | \ZR{i}, \ZR{i'}) &=& \alphar{\ZR{i}\ZR{i'}}.
    \end{array}
\end{equation}
As a consequence, the blocks gather nodes (blocks of individuals on the one hand and blocks of organizations on the other hand) sharing the same profiles of connectivity. \\
In order to take into account the fact that organizations may shape the individual behaviors, we assume that   the   memberships of the individuals ($\ZR{}$) depend on the blocks of the organizations ($\ZL{}$) they are affiliated to. More precisely, we set:
\begin{equation}\label{eq:modZI}
    \mathbb{P}(\ZR{i} = k |  \ZL{j}, \aff_{ij} = 1 ) = \gamma_{k\ZL{j}} 
    \quad \forall i \in \{1, \dots, \nbr \}
    \quad \forall k \in \{1, \dots, \QR \} ,
\end{equation}
where $\gamma$ is a $\QR \times \QL$ matrix such that $\sum_{k=1}^{\QR} \gamma_{kl} = 1$  $\forall l \in \{1,\dots,\QL\}$. 
The $(\ZL{j})$ are assumed to be independent random variables distributed as 
\begin{equation}\label{eq:modZO}
    \mathbb{P}(\ZL{j} = l) = \pil{l}, \quad 
    \quad \forall j \in \{1, \dots, \nbl \}
    \quad \forall l \in \{1, \dots, \QL \},
\end{equation}
with $\sum_{l=1}^{\QL} \pil{l}= 1$. \\
\noindent Equations \eqref{eq:modZI} and \eqref{eq:modZO} state that  the clustering of an individual is not completely driven by his/her behavior but is also shaped by the clustering of the organization he/she belongs to.
In particular, if $\QL = \QR$ and $\gamma$ is equal to the identity matrix (up to a reordering of the rows) then, the clustering of the individuals is completely determined by the clustering of the organizations. At the opposite, if all the columns of $\gamma$ are equal, then the clustering of the individuals is independent on the clustering of the organizations. This point will be developed hereafter. \\

Equations  \eqref{eq:connec}, \eqref{eq:modZI} and  \eqref{eq:modZO}  define a joint modeling of  $\XR{}$ and $\XL{}$.  In what follows, we set $\theta = \{\pil{}, \gamma\,, \alphal{}, \alphar{} \}$ the vector of the unknown parameters, $\obs = \{\XR{}, \XL{}\}$ are the observed variables and $\lat = \{\ZR{}, \ZL{}\}$ the latent variables.  The DAG of the MLVSBM is plotted in Figure \ref{fig:dag}. An illustration of the MLVSBM for a small multilevel network is represented in Figure \ref{fig:vuenetwork}.

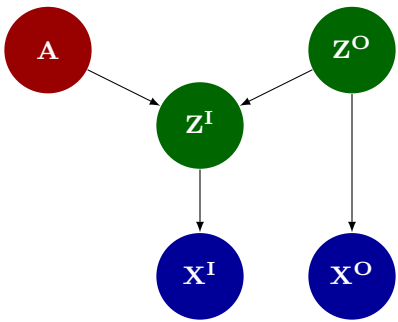
\begin{figure}
    \centering
\begin{tikzpicture}
	\tikzstyle{every edge}=[-,>=stealth',shorten >=1pt,auto,thin,draw]
	\tikzstyle{every state}=[draw=none,text=white,scale=1.2, font=\scriptsize, transform shape]
	\tikzstyle{every node}=[fill=green!50!black]
	\node[state,fill=red!60!black] (XRL) at (-2,0) {$\mathbf{\aff}$};
	\node[state,fill=green!40!black] (ZL) at (2,0) {$\mathbf{\ZL{}}$};
	\node[state,fill=green!40!black] (ZR) at (0,-1) {$\mathbf{\ZR{}}$};
	\node[state,fill=blue!60!black] (XR) at (0,-3) {$\mathbf{\XR{}}$};
	\node[state,fill=blue!60!black] (XL) at (2,-3) {$\mathbf{\XL{}}$};

	\draw[->,>=latex] (XRL) -- (ZR);
	\draw[->,>=latex] (ZR) -- (XR);
	\draw[->,>=latex] (ZL) -- (ZR);
	\draw[->,>=latex] (ZL) -- (XL);
\end{tikzpicture}
    \caption{DAG of the stochastic block model for multilevel network (MLVSBM)}
    \label{fig:dag}
\end{figure}

\paragraph{Likelihood}
From Equations \eqref{eq:connec}, \eqref{eq:modZI} and  \eqref{eq:modZO}, we  derive the complete log-likelihood for  an  undirected MLVSBM:
\begin{align}\label{eq:complelik}
	\log \ell_{\theta} &\left( \XR{}, \XL{}, \lat | \aff \right) =  \log \ell_{\pil{}}(\ZL{})+
	\log\ell_{\gamma} (\ZR{} | \ZL{}, \aff)  
	+\log \ell_{\alphar{}} ( \XR{} | \ZR{})+\log \ell_{\alphal{}} ( \XL{} | \ZL{}) \nonumber \\
	&= 
     \sum_{j,l} \mathds{1}_{\ZL{j}=l}\log\pil{l} + \sum_{i,k} \mathds{1}_{\ZR{i}=k} \sum_{j,l} \aff_{ij} \mathds{1}_{\ZL{j}=l}\log \gamma_{kl}\\
		& \;  + \frac{1}{2} \sum_{i' \neq i} \sum_{k, k'} 
			\mathds{1}_{\ZR{i}=k}\mathds{1}_{\ZR{i'}=k'} \log \phi(\XR{ii'}, \alphar{kk'})
		 + \frac{1}{2} \sum_{j' \neq j} \sum_{l, l'} 
			\mathds{1}_{\ZL{j}=l}\mathds{1}_{\ZL{j'}=l'} \log \phi(\XL{jj'} ,\alphal{ll'}), \nonumber
\end{align}

where $\phi (x, a) = a^x(1-a)^{1-x}$.  
\begin{remark}
Note that the factors $1/2$  in Equation \eqref{eq:complelik} derive from the fact that we consider undirected networks.  If one or both of the networks are directed, then the corresponding $1/2$ disappears. \end{remark}

The log-likelihood of the observations $\ell_{\theta}(\obs | A)$ is obtained by integrating out the latent variables $\lat$ in Equation \ref{eq:complelik}. As soon as $\nbl$, $\nbr$, $\QL$, or $\QR$ increase, this summation over all the possible clusterings $\ZR{}$ and $\ZL{}$ cannot be performed within a reasonable computational time. As a consequence, we will resort to the variational EM algorithm to maximize this likelihood (see Section \ref{sec:inf}).

\paragraph{Independence}
We now derive conditions for the structural independence between levels in terms of parameters equality.

\begin{proposition}\label{prop:independence}
In the MLVSBM, the two following properties are equivalent: 
\begin{enumerate}
  \item $\ZR{}$ is independent on  $\ZL{}$, 
  \item $\gamma_{kl} = \gamma_{kl'} \quad \forall l, l' \in \{1, \dots, \QL\}$,
\end{enumerate}
and imply that:  
\begin{enumerate}
  \item[3.] $\XR{}$ and $\XL{}$ are independent.
\end{enumerate}

\end{proposition}
This proposition is proved   in \ref{appendix:independence}.
Proposition  \ref{prop:independence} can be interpreted as follows: in the case where the clustering of the individuals does not depend on the clustering of the organizations, all column vectors of $\gamma$ are identical. Hence, under this restriction on $\gamma$, the model for multilevel network can be rewritten as the product of two independent SBMs, one for each level. 
Conversely, in the case of a strong dependence between the levels, each column of $\gamma$ will have one coefficient  close to one, the others being close to $0$. Therefore,  the individuals  affiliated to organizations belonging to the  same block of organizations will be affiliated to one block of individuals. Even if the $\gamma$'s imply a dependent relationship between the two levels, the connections of the corresponding blocks at the two levels may have different connectivity patterns since there is no constraint on the corresponding connection parameters $\alpha^O$ and $\alpha^I$.

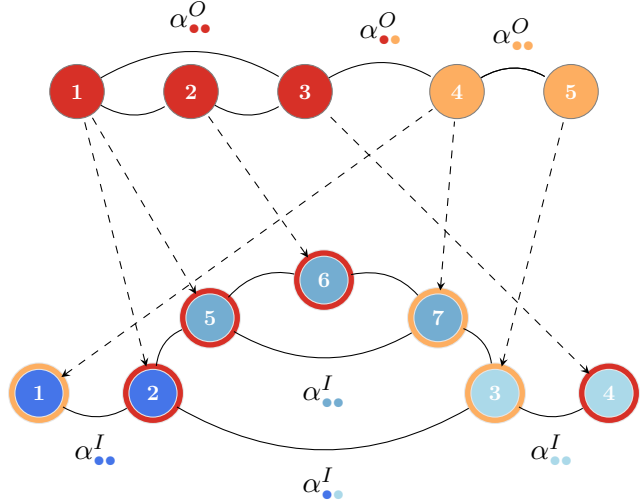
\begin{figure}
    \centering
    \begin{tikzpicture}
          \node[fill=white, scale = 1] at (3, 4.5) {};  
          \tikzstyle{every edge}=[-,>=stealth',shorten >=1pt,auto,thin,draw]
          \tikzstyle{every state}=[draw, text=white,scale=0.95, transform shape]
    
          \tikzstyle{every state}=[draw=none,text=white,scale=0.75,  transform shape]

          \tikzstyle{every node}=[fill=redorg]
          \node[state, draw=black!50] (LA1) at (0.5,3) {\textbf{1}};
          \node[state, draw=black!50] (LA2) at (2,3) {\textbf{2}};
          \node[state, draw=black!50] (LA3) at (3.5,3) {\textbf{3}}; 
          \tikzstyle{every node}=[fill=orangeorg]
          \node[state, draw=black!50] (LB1) at (5.5,3) {\textbf{4}};
          \node[state, draw=black!50] (LB2) at (7,3) {\textbf{5}};

          \path (LB1) edge [bend left] (LB2);
		  \path (LB1) edge [bend left] node[fill=white, above=.1cm]
		  {$\alphal{ \textcolor{orangeorg}{\bullet}\textcolor{orangeorg}{\bullet}}$}
		  (LB2);
		  \path (LA3) edge [bend left] node[fill=white, above=.1cm]
		  {$\alphal{ \textcolor{redorg}{\bullet}\textcolor{orangeorg}{\bullet}}$}
		  (LB1);
 		  \path (LA1) edge [bend right] (LA2);
		  \path (LA2) edge [bend right] (LA3);
		  \path (LA1) edge [bend left] node[fill=white, above=.1cm]
		  {$\alphal{ \textcolor{redorg}{\bullet}\textcolor{redorg}{\bullet}}$}
		  (LA3);
          
                    
          \tikzstyle{every node}=[fill={blueind},  double distance = 2pt]
          \node[state, draw=black!10, double=orangeorg] (RA1) at (0, -1) {\textbf{1}};
          \node[state, draw=black!10, double=redorg] (RA2) at (1.5, -1) {\textbf{2}};
          \tikzstyle{every node}=[fill={greenind}, double distance = 2pt]
     	  \node[state,draw=black!10, double=orangeorg] (RB1) at (6, -1) {\textbf{3}};
          \node[state,draw=black!10, double=redorg] (RB2) at (7.5, -1) {\textbf{4}};
          \tikzstyle{every node}=[fill=cyanind,  double distance = 2pt]
          \node[state,draw=black!10, double=redorg] (RC1) at (2.25,0) {\textbf{5}};
          \node[state,draw=black!10, double=redorg] (RC2) at (3.75,.5) {\textbf{6}};
          \node[state,draw=black!10, double=orangeorg] (RC3) at (5.25,0) {\textbf{7}};

          \tikzstyle{every edge}=[-,>=stealth', shorten >= 1pt, auto, thin, draw]
          \path (RA1) edge [bend right] node[fill=white,below=.1cm]
          {$\alphar{\textcolor{blueind}{\bullet}\textcolor{blueind}{\bullet}}$}
          (RA2);
          \path    (RB1)    edge     [bend    right]    node[fill=white, below=.1cm]
          {$\alphar{\textcolor{greenind}{\bullet}\textcolor{greenind}{\bullet}}$}  (RB2) ;
          \path 	(RC1) edge [bend right] node[fill=white,below=.1cm]
          {$\alphar{\textcolor{cyanind}{\bullet}\textcolor{cyanind}{\bullet}}$} (RC3)
            (RC1) edge [bend left] (RC2)
          	(RC2) edge [bend left] (RC3);
          \path 
          (RA2)    edge    [bend    right]    node[fill=white, below=.1cm]
	         {$\alphar{\textcolor{blueind}{\bullet}\textcolor{greenind}{\bullet}}$} (RB1)
          (RC3) edge [bend left] (RB1)
          (RA2) edge [bend left] (RC1);
          
          
		  \tikzstyle{every edge}=[dashed,>=stealth, <-,shorten >=1pt,auto,thin,draw]
		  \path (RA1) edge  (LB1);
		  \path (RA2) edge  (LA1);
		  \path (RC1) edge  (LA1);
		  \path (RC2) edge  (LA2);
		  \path (RC3) edge  (LB1);
		  \path (RB1) edge  (LB2);
		  \path (RB2) edge  (LA3);
          
        \end{tikzpicture}

    \caption{MLVSBM with inter-organizational level on the top and inter-individual level on the bottom.  The various shades of blue depict the clustering of the individuals and the various shades of red depict the clustering of the organizations. The parameters $\alpha$ over the plain links between nodes are the probabilities of connections given the nodes colors (clustering/blocks). The outer circles around the nodes of the individuals represent the blocks of the organizations they are affiliated to. The dashed links stand for the affiliations.}
    \label{fig:vuenetwork}
\end{figure}

\paragraph{Identifiability}
The identifiability conditions for the MLVSBM are given in the following proposition. 

\begin{proposition}\label{prop:identif}
The MLVSBM is identifiable up to label switching under the following assumptions:
\begin{enumerate}
\item[$\mathcal{A}$1.] All coefficients of  $\alphar{}\cdot\gamma\cdot\pil{}$ are distinct and all coefficients of $\alphal{}\cdot\pil{}$ are distinct.
\item[$\mathcal{A}$2.] $\nbr \geq 2\QR$ and $\nbl \geq \max(2\QL, \QL+\QR-1)$.
\item[$\mathcal{A}$3.] At least $2\QR$ organizations contain one individual or more.
\end{enumerate}
\end{proposition}
\noindent The set of parameters that does not verify  assumption $\mathcal{A}$1 has null Lebesgue measure. 

Assumption $\mathcal{A}$2 is very weak in practice.  Assumption $\mathcal{A}$3, on the affiliation, means that at least some organizations must not be empty and  
enough individuals belong to different organizations. 
The proof of this proposition is provided  in  \ref{appendix:identifiability}  and results from an extension of the proof given in \citet{celisse2012consistency}.

\section{Statistical Inference}\label{sec:inf}
We now present a maximum likelihood procedure and a criterion for model selection.

\subsection{Variational method for maximum likelihood estimation}

As said before,    $\ell_{\theta} (\obs | \aff)$ is obtained by integrating out  the latent variables $\lat$ in the complete data likelihood \eqref{eq:complelik}. 
However, this calculus  becomes not computationally tractable as the numbers of nodes and blocks increase.

The Expectation-Maximization algorithm (EM) \citep{dempster1977maximum} is a popular solution to maximize the likelihood of models with latent variables. However it requires the computation of $\mathbb{P}_{\theta}(\lat | \obs, \aff)$ which is also not tractable in our case. The variational version of the EM algorithm is a powerful solution for such cases. It was first used for the SBM by  \cite{daudin2008mixture}. 

In a few words, the variational EM algorithm  maximizes the so-called variational bound i.e. a lower bound of the log-likelihood  denoted   $\mathcal{I}_{\theta}(\mathcal{R}(\lat | \aff))$  and defined as follows:
\begin{eqnarray}\label{eq:vbound}
 \mathcal{I}_{\theta}(\mathcal{R}(\lat | \aff)) & :=&  \mathbb{E}_\mathcal{R} \left[ \ell_{\theta} \left( \lat, \obs | \aff \right) \right] + \mathcal{H}\left( \mathcal{R}(\lat | \aff) \right) \\ 
 & =&  \ell_{\theta} \left( \obs | \aff \right) - \KL \left(\mathcal{R}(\lat | \aff) \| \mathbb{P}_{\theta}(\lat | \obs, \aff) \right) \leq	\ell_{\theta} (\obs | \aff), \nonumber
\end{eqnarray}
where $\KL$ is the Kullback-Leibler divergence, $\mathcal{H}$ is the Shannon entropy: $\mathcal{H}(P) = \mathbb{E}_{P}[-\log(P)]$ and
$\mathcal{R}(\lat | \aff)$ is an approximation of the true distribution   $\mathbb{P}_{\theta}(\lat | \obs, \aff)$. In our context, and following \cite{daudin2008mixture}, we propose to choose $\mathcal{R}(\lat | \aff)$ in  a family of factorized distributions, resulting into a mean field approximation  $\mathcal{R}(\lat | \aff)$ defined as:
\begin{equation}\label{eq:R}
    \mathcal{R}(\lat | \aff) = \prod_{i=1}^{\nbr} \prod_{k=1}^{\QR}(\tau^{I}_{ik})^{\ind_{\ZR{i} = k}} \prod_{j=1}^{\nbl} \prod_{l=1}^{\QL}(\tau^{O}_{jl})^{\ind_{\ZL{j} = l}} , 
\end{equation}
where  $ \tau^I_{ik} = \mathbb{P}_{ \mathcal{R}}(\ZR{i} = k)$ and $ \tau^O_{jl} = \mathbb{P}_{ \mathcal{R}}(\ZL{j} = l)$. \\
Inputting Equations \eqref{eq:complelik} and \eqref{eq:R} into Equation \eqref{eq:vbound}, the variational bound for the MLVSBM can be written as follows:
\begin{eqnarray*}
    \mathcal{I}_{\theta}(\mathcal{R}(\lat | \aff)) & =& 
     \sum_{j,l} \taul{jl}\log\pil{l}   + \sum_{i,k} \taur{ik} \sum_{j,l} \aff_{ij} \taul{jl}\log \gamma_{kl} \\
		&& + \quad \frac{1}{2} \sum_{i' \neq i} \sum_{k, k'} 
			\taur{ik}\taur{i'k'} \log \phi \left( \XR{ii'} , \alphar{kk'}\right) + \frac{1}{2} \sum_{j' \neq j} \sum_{l, l'} 
			\taul{jl}\taul{j'l'} \log \phi\left( \XL{jj'} \alphal{ll'} \right)\\
			&& - \sum_{i,k} \taur{ik}\log \taur{ik} - \sum_{j,l} \taul{jl}\log \taul{jl}.
\end{eqnarray*}
The variational EM algorithm consists in iterating two steps. Step \textsf{VE}  maximizes the variational bound with respect to the parameters of the approximate distribution defined in Equation \eqref{eq:R}. This is equivalent to minimizing the Kullbach-Leibler divergence term. Step \textsf{M} maximizes the variational bound with respect to the model parameters $\theta$. The procedure is given in Algorithm \ref{algo:vem} and details of the calculus and algorithm are developed in \ref{appendix:vem}.
Algorithm \ref{algo:vem} can be slightly modified to handle missing data (\OA{}{dyads which are not observed in any of the two levels}) by summing up on observed dyads only.
An interesting feature of the MLVSBM is to make use of one level to help the prediction of missing dyads of the other level.
\begin{remark}
  \OA{}{Although the family of the variational distributions does not consider the affiliation matrix $\aff{}$, the minimization of the Kullback-Leibler divergence between the variational distribution and $\mathbb{P}_{\theta}(\lat | \obs, \aff)$ induces an indirect dependence on $A$ in the variational distribution.
  One may consider more complex distributions but the simulation studies show that the inference algorithm is able to retrieve properly the dependence between the $\ZR{}$s and the $\ZL{}$s in this family of distributions.}
\end{remark}

\begin{algorithm}[hbt!]\label{algo:vem}
\KwData{$\{\obs, \lat, \aff \}$, a multilevel network with an initial clustering of size $(\QR, \QL)$.}
\BlankLine
\textbf{Procedure:}\\

$\bullet$ Set $\{\taur{}, \taul{}\}$ from the initial clustering.\\
\vspace{0.5em}
\While{$\mathcal{I}_{\theta}(\mathcal{R}(\lat | \aff))$ is increasing}{
\begin{description}
    \item[$\bullet$ M step] compute 
    \[\theta^{(t+1)} = \arg \underset{\theta}{\max}\;  \mathcal{I}_{\theta}(\mathcal{R}^{(t+1)}(\ZR{}, \ZL{} | \aff)),\]
    by updating the model parameters as follows:
    \begin{align*}
	\widehat{\pil{l}} & =  \frac{1}{\nbl} \sum_j \widehat{\taul{jl}} & 
	\widehat{\alphal{ll'}} & = \frac{ \sum_{j' \neq j} \widehat{\taul{jl}}\XL{jj'}\widehat{\taul{j'l'}}}{\sum_{j' \neq j} \widehat{\taur{jl}}  \widehat{\taur{j'l'}} } \\
    \widehat{\gamma}_{kl} & = \frac{ \sum_{i,j}  \widehat{\taur{ik}} A_{ij} \widehat{\taul{jl}} }{ \sum_{i, j} A_{ij}\widehat{\taul{jl}} } &
	\widehat{\alphar{kk'}} & = \frac{ \sum_{i' \neq i} \widehat{\taur{ik}} \XR{ii'} \widehat{\taur{i'k'}}}{\sum_{i' \neq i} \widehat{\taur{ik}} \widehat{\taur{i'k'}} }\,.
	\end{align*}
    \item[$\bullet$ VE step] compute 
    \begin{align*} \{\taur{}, \taul{} \}^{(t+1)} & = \arg \underset{\taur{}, \taul{}}{\max}\; \mathcal{I}_{\theta^{(t)}}(\mathcal{R}(\ZR{}, \ZL{} | \aff))
    \end{align*}
    by updating the variational parameters with the following fixed points relationships:
    \begin{align*}
	\widehat{\taul{jl}}  \propto & \pil{l} \prod_{i,k} \gamma_{kl}^{\aff_{il}\widehat{\taur{ik}}}\prod_{j'\neq j}\prod_{l'}\phi(\XL{jj'}, \alphal{ll'})^{\widehat{\taul{j'l'}}} \\
    \widehat{\taur{ik}}  \propto &  \prod_{j,l} \gamma_{kl}^{\aff_{il}\widehat{\taul{jl}}}\prod_{i'\neq i}\prod_{k'}\phi(\XR{ii'}, \alphar{kk'})^{\widehat{\taur{i'k'}}}\,.
\end{align*}
\end{description}
}
 \Return{$\mathcal{I}_{\theta}(\mathcal{R}(\lat | \aff))$, $\widehat{\theta}$ and $\{\widehat{\taur{}}, \widehat{\taul{}}\}$}
 
\caption{Variational EM algorithm}
\end{algorithm}

\subsection{Model selection}
\subsubsection{Selection of the number of blocks}

Following \cite{biernacki2000assessing} and \cite{daudin2008mixture},  we propose  a model selection criterion to choose the unknown  number of blocks $\QR$ and $\QL$. 
 The ICL criterion is an
 integrated  version of BIC applied to the complete likelihood.
 In other words, it is an asymptotic approximation of the complete likelihood integrated over its parameters and latent variables, it values both goodness of fit and classification sharpness \citep{mariadassou2010uncovering}.\\

Our criterion is equal to:

\begin{equation}\label{eq:icl}
\ICL_\MLVSBM(\QR, \QL) = \log \ell_{\widehat{\theta}}(\XR{}, \XL{}, \widehat{\ZR{}}, \widehat{\ZL{}} | A, \QR, \QL) - \pen_\MLVSBM(\QR, \QL),
\end{equation}
where
\begin{align}\label{eq:penalty}
    \pen_\MLVSBM(\QR, \QL)  = & \frac{1}{2}\frac{\QR(\QR+1)}{2}\log\frac{\nbr(\nbr-1)}{2} + \frac{\QL(\QR-1)}{2} \log \nbr + \nonumber \\
    & \frac{1}{2}\frac{\QL(\QL+1)}{2}\log\frac{\nbl(\nbl-1)}{2} + \frac{\QL-1}{2} \log \nbl ,
\end{align}
where $\widehat{\ZL{}}$ and $\widehat{\ZR{}}$ are the imputed latent variables  using the maximum a posteriori (MAP) of $ \mathbb{P}_{\hat{\theta}}(\lat | \obs, \aff ; \QR, \QL)$.  
The calculus  is provided in \ref{appendix:icl}. 
As for the variational inference, $ \mathbb{P}_{\hat{\theta}}(\lat | \obs, \aff ; \QR, \QL)$  is unknown and, in practice, we replace it by its mean-field approximation $\mathcal{R}_{\hat{\theta}}(\lat | \aff; \QR, \QL)$.

\begin{remark}
Once again, note that the penalty \eqref{eq:penalty} is adapted to  undirected networks. For instance, the term $\frac{\QR(\QR+1)}{2}\log\frac{\nbr(\nbr-1)}{2}$ would become $\QR^2 \log \nbr(\nbr-1)$ if $\XR{}$ 
were not symmetric.   
\end{remark}
\begin{remark}
We recall that the penalty of the ICL for a (unilevel) SBM is given by

\begin{align}\label{eq:penaltysbm}
    \pen_{\SBM}(Q)  = & \frac{1}{2}\frac{Q(Q+1)}{2}\log\frac{n(n-1)}{2} + \frac{Q-1}{2} \log n. 
\end{align}
\end{remark}
The penalty term in Equation \eqref{eq:penalty} for the inter-organizational level is the same as the one given in Equation \eqref{eq:penaltysbm}. For the inter-individual network, the factor in front of $\log{\nbr{}}$ is  $\QL(\QR-1)$ instead of $\QR-1$ for the SBM as in Equation \eqref{eq:penaltysbm}, that is the penalty term which corresponds to the degree of freedom of $\gamma$.

\subsubsection{Determining the independence between levels}

The ICL criterion can also be used to assess whether the two levels of interactions are independent or not. If  $\gamma$  is forced to have all its   columns identical, then the penalty term on $\gamma$ becomes $\frac{1}{2}(\QR-1) \log \nbr$ and, as a consequence: 
\begin{equation}\label{eq:ICL ind}
    \ICL_{\Ind}(\QR, \QL) = \ICL^I_{\SBM}(\QR) + \ICL^O_{\SBM}(\QL).
\end{equation}
The ICL criterion favors independence if 
\begin{equation*}
\underset{\{\QR, \QL\}}{\max} \ICL_{{\MLVSBM}} (\QR, \QL) \leq  \underset{\QR}{\max}\; \ICL^I_{\SBM}(\QR) + \underset{\QL}{\max}\; \ICL^O_{\SBM}(\QL).
\end{equation*}
If this is the case, then the gain in terms of likelihood does not compensate the gain $\frac{1}{2}(\QL-1)(\QR-1) \log \nbr$ in the penalty. This criterion focuses on the dependence between levels given by the inter-level. 

\begin{remark}
  If  $\QR=1$ or $\QL=1$, the MLVSBM is the product of two independents SBM, as such $\ICL_{\Ind}(\QR, \QL) = \ICL_{{\MLVSBM}} (\QR, \QL)$.
\end{remark}

\subsubsection{Procedure for model selection}
We now provide a procedure for model selection which seeks for the optimal number of blocks at a reasonable cost. As a by-product, it states whether the two levels are independent or not.\\

The practical choice  of the model and the estimation of its parameters are computationally intensive tasks. Indeed,  we should  compare all the possible models -- one model corresponding to a given $(\QR, \QL)$ -- through the ICL criterion. Furthermore, for each model, the variational EM algorithm  should be  initialized  at a large number of initialization points  (due to its sensitivity to the starting point), resulting in an unreasonable computational cost.  Instead, we propose to adopt a  stepwise strategy, resulting in a faster exploration of the model space,  combined with  efficient initializations of the variational EM algorithm.  The procedure we suggest is  given in Algorithm  \ref{algo:icl}.

\begin{algorithm}[hbt!]\label{algo:icl}
\KwData{$\{\XR{}, \XL{}, \aff\}$, a multilevel network.}
\BlankLine
\textbf{Procedure}:\\

$\bullet$ Infer independent SBMs on $\XR{}$ and $\XL{}$
for a respective range of $\QR$ and $\QL$. Deduce 
\begin{eqnarray*}
 \widehat{\QR}^{\Ind} &=& \argmax_{\QR} {\ICL^I_{\SBM}}(\QR)\quad \mbox{ and }  \quad \widehat{\QL}^{\Ind} = \argmax_{\QL} {\ICL^O_{\SBM}}(\QL).
\end{eqnarray*}
 Compute  $\ICL_{\Ind}  = \ICL^I_{\SBM}(\widehat{\QR}^{\Ind} ) + \ICL^O_{\SBM}(\widehat{\QL}^{\Ind})$.\\
\vspace{0.5em}

$\bullet$ Start at $  \QR = \widehat{\QR}^{\Ind}$ and $\QL = \widehat{\QL}^{\Ind} $.\\
\vspace{0.5em}
\While{ICL is increasing}{
	- Fit an MLVSBM on every model of size $(\QR \pm 1, \QL \pm 1)$ initialized by merging 2 blocks or splitting a block with hierarchical clustering. \\
	- Among all estimated models, keep the one with the highest ICL.
}
 \Return{$(\widehat{\QR}, \widehat{\QL}) = \argmax \ICL(\QR,\QL)$, $\widehat{\theta}_{(\widehat{\QR}, \widehat{\QL})}$ and $\widehat{\lat}$.}

\caption{Model selection algorithm}
\end{algorithm}

Each step of the algorithm requires $O(\max \{\QR, \QL \}^2)$ variational EM algorithms which converge in a few iterations as a result of the local initialization. Inferring an independent SBM on each level beforehand is a fast way to start with good initialization and allows us to state on the independence of the model at the same time as we just need to compare the sum of the $ICL_\Ind$  and $\ICL_\MLVSBM(\widehat{\QR}, \widehat{\QL})$. 

\paragraph{Package}

All the codes are available as an \texttt{R} package  at \href{https://chabert-liddell.github.io/MLVSBM/}{https://chabert-liddell.github.io/MLVSBM/}.
It features the simulation and inference of multilevel networks with  symmetric and/or asymmetric adjacency matrices, model and independence selection. It also handles missing at random data \citep{rubin1976inference} on the adjacency matrices \OA{}{of one or both levels }
and link prediction.

\section{Illustration on simulated data}\label{section:simulation}
In this section, we study the performances of the
inference procedure for the MLVSBM including the ability to recover blocks, the selection of the numbers of blocks and the independence detection.

\begin{remark}
In order to evaluate the ability to recover blocks, we resort to the Adjusted Rand Index (ARI) \citep{hubert1985comparing} which is a comparison index between two clusterings with a correction for chance. This index is close to $0$ when the two clusterings are independent and is $1$ when the clusterings are identical (up to label switching).  
\end{remark}

\begin{remark}
\OA{}{In our results, we focus on the ability to recover blocks rather than on the quality of the model parameter estimates since it is the hardest task.
Indeed, once the blocks are recovered (ARI=1), the estimation of the model parameters boils down to the computation of the proportions of observed links between blocks which is a consistent estimator in a Bernoulli i.i.d. model.}
\end{remark}

\subsection{Experimental design}
In what follows, we set $\QL = \QR = 3$. The networks are of sizes:    $\nbl = 60 $   and $\nbr = 180$. 

Let $d$ be a density parameter: the lower $d$, the sparser the network and the harder the inference.  $\epsilon$ ($\geq 1$) is a parameter tuning the strength of the communities;  when $\epsilon$ is high, the communities are easily separable. 
In the simulation study, we focus on the three following standard topologies. 
\begin{itemize}
    \item  \emph{Assortative communities.} The probability of connection within communities  is higher than the probability of connection between communities:
   $\alphar{} = d*
        \begin{bmatrix}
        \epsilon & 1 & 1 \\
        1 & \epsilon & 1 \\
        1 & 1 & \epsilon
        \end{bmatrix}$. 
    \item  \emph{Disassortative communities.} The probability of connection within communities is lower than the probability of connection between communities:
   $\alphar{} = d*\begin{bmatrix}
        1 & \epsilon & \epsilon \\
        \epsilon & 1 & \epsilon \\
        \epsilon & \epsilon & 1
        \end{bmatrix}$. 
    \item  \emph{Core-periphery.} A core block is highly connected to the whole network while the probability of connection in the periphery is low: 
   $\alphar{} = d*
        \begin{bmatrix}
        \epsilon & \epsilon & 1 \\
        \epsilon & 1 & 1 \\
        1 & 1 & 1
        \end{bmatrix}$. 
\end{itemize}

\noindent We fix the topology of the inter-organizational level $\XL{}$ to be an assortative communities with $d = 0.1$,  $\epsilon = 5$ and of communities of equal size on average. We expect this topology to be easy to infer and to obtain a perfect recovery of the clustering with high probability.

 For the inter-individual level, $d$ is set  to $0.01$, $0.05$ or $0.1$ while $\epsilon$ ranges from 1 to 10 by stepize of 0.5.  $\epsilon = 1$  corresponds to an Erd\H{o}s-Rényi graph and the communities should be indistinguishable.   \\The affiliation matrix $A$ is generated from a power-law distribution in order to get different sizes of organizations. Other distributions were tried but the results (not reported here) show that their impact  on the inference is weak. 
 
Finally, $\delta$ is a parameter for the strength of the dependence between levels,   ranging from $0$ to $1$.
More precisely, we set: 
\begin{equation*}
    \gamma = 
    \begin{bmatrix}
        \delta & \frac{1}{2}(1-\delta) & \frac{1}{2}(1-\delta)\\
        \frac{1}{2}(1-\delta) & \delta & \frac{1}{2}(1-\delta)\\
        \frac{1}{2}(1-\delta) & \frac{1}{2}(1-\delta) & \delta\\
    \end{bmatrix}
\end{equation*}
where $\gamma$ has been defined in Equation \eqref{eq:modZI}.
$\delta = 1/\QR$  corresponds to the case of independence between levels. The further $\delta$ is from $1/\QR$, the stronger the dependence between levels. $\delta = 1$ implies a deterministic link between the clustering of the two levels, ie. the block of an individual is fully determined by the block of his/her organization. 
\OA{}{With this experimental design we aim to exhibit how the inference is improved by applying the MLVSBM rather than the SBM when the two levels are intertwined.}

\subsection{Simulation results}
\OA{}{During the inference procedure, the number of blocks is unknown for both levels.
We run the model selection for $\widehat{\QR}\in\{1,\ldots,10\}$ and $\widehat{\QL}\in\{1,\ldots,10\}$.}

First, we fix $\delta = 0.8$ and make $\epsilon$ vary. Each situation is simulated 50 times.
We  test  the ability of our model to recover the true clustering of $\ZR{}$ from $(\XR{},\XL{})$. We compare our performances to the ones obtained by applying a standard (unilevel) SBM on $\XR{}$.   
Because  $(\QR,\QL)$  are  assumed to be unknown,  two types of error may occur:  one for not selecting the right $\QR$ and one for assigning nodes to the wrong blocks.
The results are displayed in Figure \ref{fig:ari_epsilon}.

 In Figure \ref{fig:ari_epsilon} A, we plot -- for $3$ values of density $d$ and the $3$ topologies (assortative, core-periphery and disassortative) -- the ARI when using MLVSBM (plain line) and SBM (dashed line) as $\varepsilon$ varies.   We observe that, for any topology, the MLVSBM starts to recover perfectly the clustering for  a lower value of $\epsilon$ than the SBM \OA{}{ because in the MLVSBM, the inter-individual level benefits from the information held in the inter-organizational level through the dependence of their blocks}. The difficulty of the inference increases as $\epsilon$ decreases: as can be seen in Figure \ref{fig:ari_epsilon} A, MLVSBM  still performs well ($ARI > 0$) for small values of $\epsilon$ while the  SBM is unable to recover the clustering.

 In Figures \ref{fig:ari_epsilon} B and C,  we plot the number of blocks chosen by the MLVSBM (B) and the SBM (C) for $3$ values of density (rows)  and $3$ topologies (columns) (the true value being $\QR{}  = 3$). 
 We observe that  using the MLVSBM allows to recover more precisely $\QR$ than using the SBM.
$\widehat{\QR}$ varies from $1$, when no structure is detected to $3$ which is the true number of blocks. The procedure never selects more blocks than expected, which   is coherent with prior  knowledge that the  ICL for  the   SBM tends to select models of smaller size   \citep{hayashi2016tractable, brault2014estimation}.\\

\begin{figure}[ht]
    \centering
    \includegraphics[width = \textwidth]{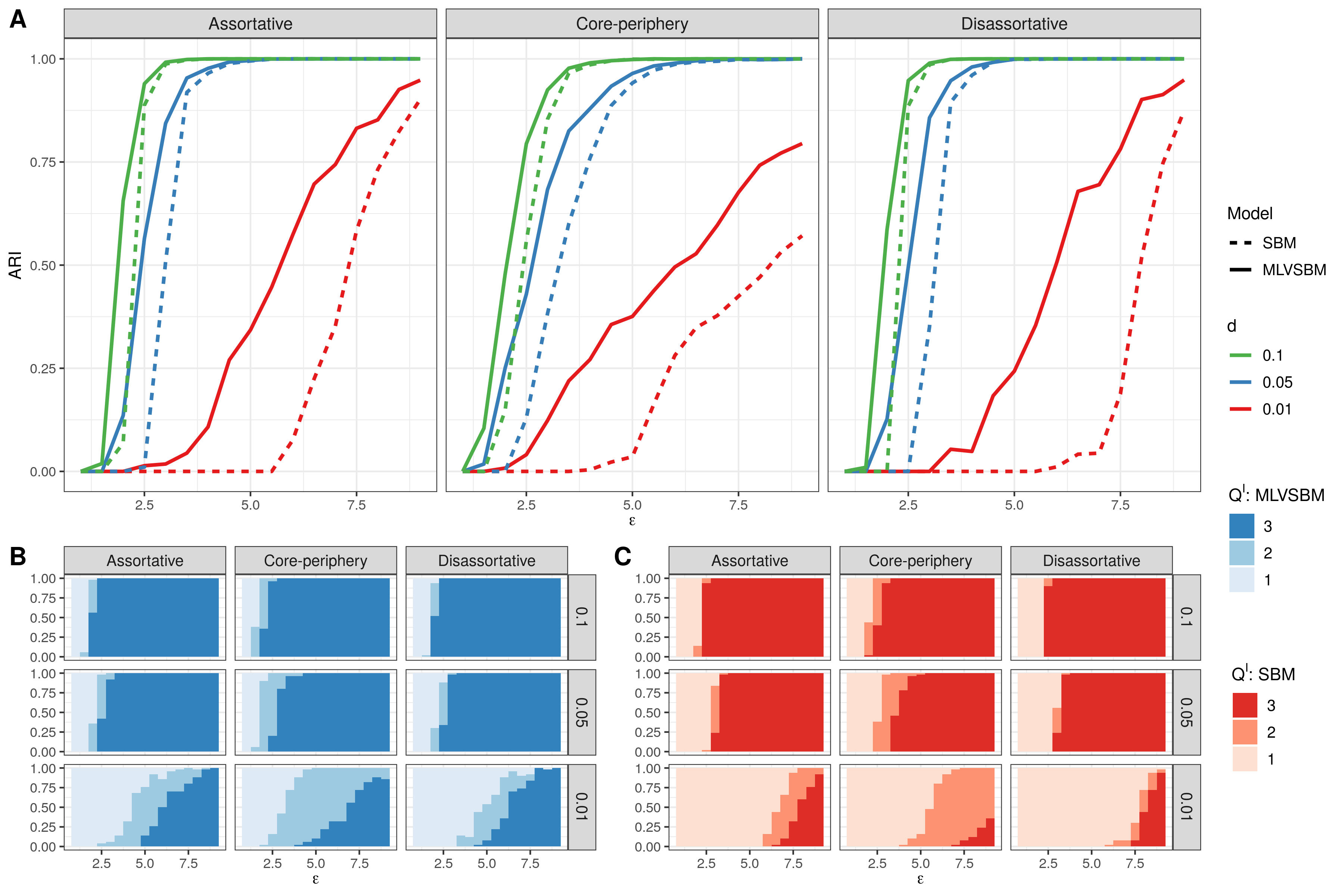}
    \caption{Clustering and model selection for 3 different topologies on the inter-individual level, varying $\epsilon$ and density $d$. Each situation is simulated 50 times. 
    \textbf{A}: ARI for the inter-individual level, comparing the model used for inference. 
    \textbf{B}: Stacked frequency barplot of the selected number of blocks for the inter-individual level in the MLVSBM (in blue).  
    \textbf{C}:  Stacked frequency barplot of the selected number of blocks for the inter-individual level chosen in the SBM (in red).}
    \label{fig:ari_epsilon}
\end{figure}

On the three topologies  with  $\epsilon = 3$,  depending on the density $d$,  MLVSBM and SBM supply $\ZR{}$ either  a perfect recovery of the clustering  or a random clustering or  something in between.  In order to understand better this phenomenon, we fix  $\epsilon$ to 3 and make $\delta$ -- which quantifies the dependency between the two levels -- vary.   
The results are reported in Figure \ref{fig:ari_gamma} for 50 simulations of each situation. 

 When $\delta=1/3$ (yellow vertical line in Figure \ref{fig:ari_gamma} A), the two levels are independent and  the results in terms of clustering  are the same for the MLVSBM and the SBM on $\XR{}$ (see ARI  in  Figure \ref{fig:ari_gamma} A).  As soon as $\delta$ departs from this value, the MLVSBM is  able to recover some of the structure of the inter-individual level thanks to the inter-organizational level and this ability is observed  even for very low density when $\delta$ gets closer to 1 (see Figure \ref{fig:ari_gamma} A and B).
 
 Figure \ref{fig:ari_gamma}.C depicts the performances of the ICL criterion  to state on the independence between the two levels. For $d = 0.01$,  $\XR{}$  is very sparse,  $\widehat{\QR} =  1$ (no structure is detected on the inter-individual level) leading to $\ICL_{ind} = \ICL_{{\MLVSBM}}$ and preventing us from detecting any dependency.  
For higher densities, we see  as expected, that if $\delta  \approx 1/3$,  the independent SBM will be preferred.  On the contrary  
the further $\delta$ departs from $1/3$ the more the MLVSBM will be selected,  even-though the MLVSBM and the independent SBM may provide the same clusterings.  This phenomenon occurs faster for higher density $d$.  
In our simulation, the MLVSBM is never selected when $\delta = 1/3$. This is a consequence of the conservative nature of ICL, requiring strong evidence from the likelihood to select a more complex model. \\

\begin{figure}[ht]
   \centering
   \includegraphics[width = \textwidth]{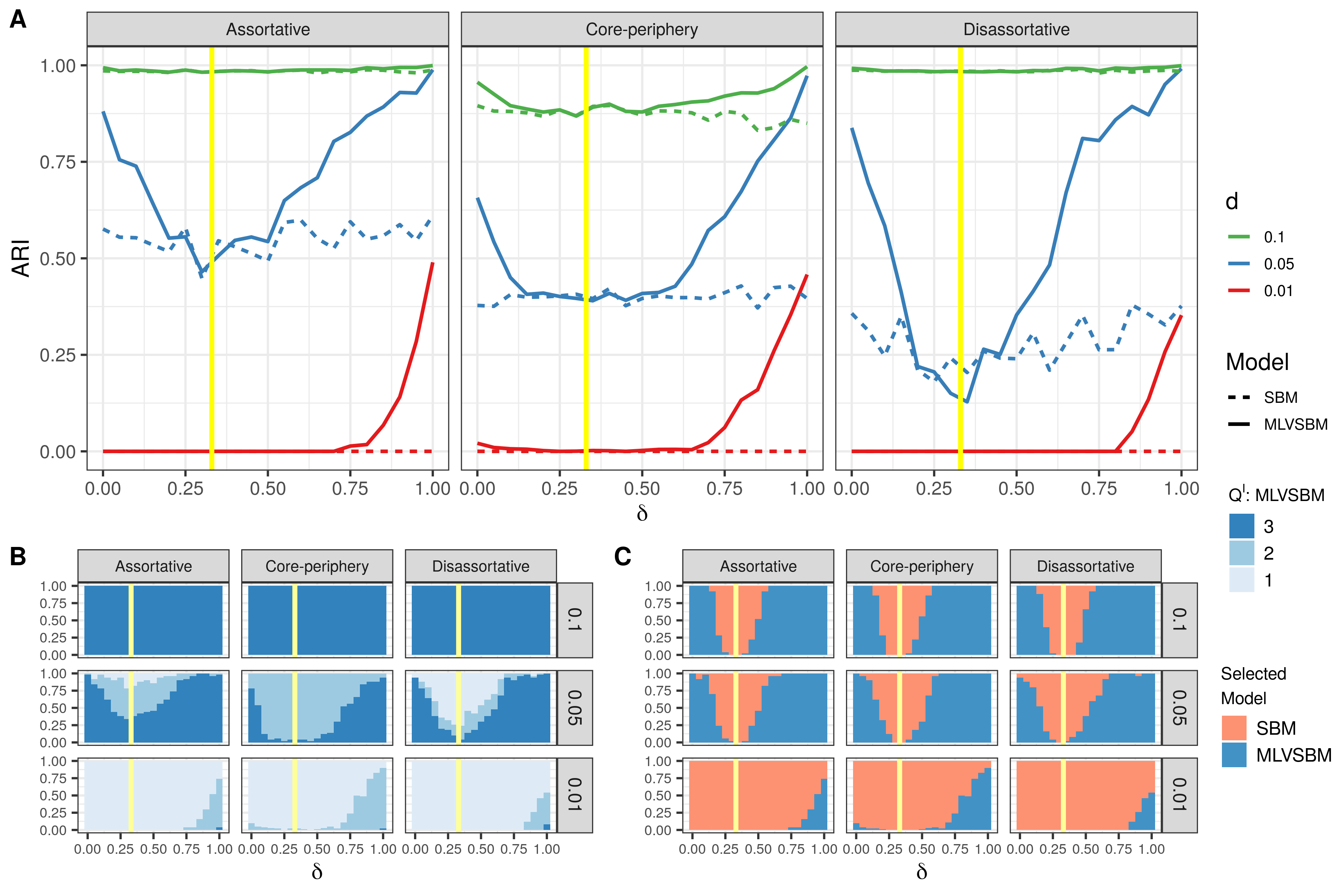}
   \caption{Clustering and model selection for 3 different topologies on the inter-individual level, as function of  $\delta$ and density $d$. Each situation is simulated 50 times. The yellow vertical lines represent a $\delta = 1/3$ (i.e. a $\gamma$ with uniform coefficients, resulting into independence between the two levels).
   \textbf{A}: ARI for the inter-individual level, comparing the model used for inference. 
   \textbf{B}: Stacked frequency barplot of the selected number of blocks for the inter-individual level in the MLVSBM. 
   \textbf{C}: Stacked frequency barplot of the selected model with respect to inter-level dependence.
   }
   \label{fig:ari_gamma}
\end{figure}

\OA{}{We chose not to present results concerning the inter-organizational level since its structure was selected to be ``easy-to-infer''. Hence, the SBM and the MLVSBM  perform well for selecting the true number of blocks $\QL$ and recovering the block structure.}
Simulations gave similar results (not reported here) when we inverse the topologies on $\XR{}$ and $\XL{}$,    showing that information on structure transits in both ways. Moreover, when the number of nodes of the "easy-to-infer" level increases, it facilitates the recovering  of the  clustering on the "hard-to-infer" level. \OA{}{When both levels are "hard-to-infer", the inference of each level benefits from one another if the dependence between the two levels is strong enough. One can exhibit cases where the unilevel SBM is unable to recover the clustering of any of the two levels but where the MLVSBM succeeds in recovering the true blocks for both. Detailed results for such a  simulation study are available on the MLVSBM \texttt{R} package website \url{https://chabert-liddell.github.io/MLVSBM/articles/hard_to_infer.html}.}

\OA{}{\subsection{Computational costs}
Inferring the blocks and the parameters of a multilevel network is a challenging task which can be time consuming. As a guideline for readers, we present in Table \ref{tab:comp_cost} the average computation time using the \texttt{R} package MLVSBM on two cores of a desktop computer with 32GB of RAM and a Intel® Xeon(R) CPU E5-1650 v4 @ 3.60GHz × 12 processor running on Ubuntu 18.04.5 LTS for the inference of simulated networks including model selection for different network sizes and different numbers of blocks.  
} 
  \begin{table}[hbt!]
    \centering
    \OA{}{  
    \begin{tabular}{ccccc}
        \toprule
        \multicolumn{2}{c}{Network Size} &
        \multicolumn{3}{c}{Running time (mean $\pm$ sd) in seconds}\\
        \cmidrule(lr){1-2}\cmidrule(lr){3-5}
        $\nbr$ & $\nbl$ & $\QR = \QL = 2$ &  $\QR = \QL = 4$ &  $\QR = \QL = 8$ \\
        \cmidrule(lr){1-2}\cmidrule(lr){3-5}
        150 & 50 & $9.87 \pm  4.13$  & -  & - \\
        600 & 200 & $443 \pm 205$   & $1794 \pm 1287$ & -   \\
        1500 & 500 & $1093 \pm 900$   & $2583 \pm 1226$  & $7050 \pm 2670$   \\
        \bottomrule
    \end{tabular}
        \caption{Average running time for the inference of the MLVSBM for different network sizes and different numbers of blocks.}
   \label{tab:comp_cost} }
\end{table}

\section{Application to the multilevel network issued from a television programs trade fair}\label{section:application}

We apply our model to the data set \citep{brailly2016embeddedness} described below.

\subsection{Context and Description of the data set}

Promoshow East is   a television programs trade fair for Eastern Europe. Sellers from Western Europe and the USA come to sell audiovisual products to regional and local buyers such as broadcasting companies.  The data gather observations on one particular  audiovisual product, namely animation and cartoons. 
From a sociological perspective, reconstituting and analyzing multilevel (inter-individual and inter-organizational) networks in this industry is important. In economic sociology, it helps redefine the nature of markets \citep{brailly2016embeddedness, brailly2017neo, lazega2002interdependent}. In the sociology of culture, it helps understand, from a structural perspective, the mechanisms underlying contemporary globalization and 
standardization of culture \citep{brailly2016embeddedness, favre2016inter}. In the sociology of organizations and collective action, it helps understand the importance of multilevel relational infrastructures for the management of tense competition and cooperation dilemmas by various categories of actors \citep{lazega2020bureaucracy}, in this case the (sophisticated) sales representatives of cultural industries.

The data were collected by face-to-face interviews. At the individual level, people were asked to select from a list the individuals from which they obtain advice or information during or before the trade fair. The level consists of $128$ individuals and $710$ directed interactions (density = $0.044$). The individuals were affiliated to $109$ organizations,  each one containing from one to  six individuals. 
At the inter-organizational level, two   kinds of interactions were collected:  a deal network (deals signed  since the last trade fair) and a meeting network (derived from the aggregation at the inter-organizational level of the meetings planned by individuals on the trade fair's website). 
Both networks are symmetric with respective densities $0.067$ and $0.059$. 

\subsection{Statistical analysis}

The MLVSBM is inferred on the two datasets (one dataset corresponding to the deal network at the inter-organizational level, the other dataset to the meeting network at the inter-organizational level).
In both cases the ICL criterion favors dependence between the two levels and chooses  $\widehat{\QR} = 4$ blocks of individuals. $\widehat{\QL} $ is equal to $3$ for the deal network  and $4$ for the meeting network.

In order to determine which is the most relevant inter-organizational network, we test the ability of the MLVSBM to predict dyads or links in the inter-individual network when the deal or the meeting networks are considered. To do so, we choose uniformly dyads and links to remove and try to predict them. 
More precisely,  we  set $\XR{ii'} = \texttt{NA}$ for a certain percentage of $(i,i')$ (this percentage ranging from $5\%$ to $40 \%$ by step-size of $5\%$).  We also propose to remove existing links (ie.  forcing $\XR{ii'} = 0$ when  $\XR{ii'} = 1$ was observed, for some randomly chosen $(i,i')$). The percentage of removed existing links varies from   $5\%$ to  $95\%$ (with step-size of $5\%$). We repeat the following procedure $100$ times:
\begin{enumerate}
    \item Remove dyads or links uniformly at random
    \item Infer the newly obtained network from scratch in order to obtain the probability of a link $ \mathbb{P} (\XR{ii'} = 1; \widehat{\theta})$ for each missing dyad or for each dyad such that $\XR{ii'} = 0$
    \item Predict link among all missing dyads or among all dyads such that $\XR{ii'} = 0$.
\end{enumerate}
 Missing data are handled as Missing At Random \citep{tabouy2019variational} and the probability of existence of an  edge is given by:
$
    \mathbb{P}(\XR{ii'}= 1; \widehat{\theta}) = \sum_{k,k'} \widehat{\taur{ik}} \widehat{\alphar{kl}} \widehat{\taur{i'k'}}
$. Since the result of our procedure is equivalent to a binary classification problem, we assess the performance through the area under the ROC curve (AUC) (a random classification corresponding to $\AUC = 0.5$). \\

Figure \ref{fig:predict_orga} shows that using the MLVSBM compared to a single level SBM improves a lot the recovery of the inter-individual level for this dataset. This confirms the dependence between levels detected by the ICL. Moreover,  using the deal network gives better predictions for both missing dyads and missing links than the meeting network.  We also considered a merged network at the inter-organizational level by making the union of links of the deal and the meeting network\OA{}{, i.e. for all $j, j' \in \nbl$, $X^{O, \text{merged}}_{jj'} = \max \{X^{O, \text{deal}}_{jj'}, X^{O, \text{meeting}}_{jj'}\} $ }. The improvement in terms of prediction over the deal network is not very significant and this composite network is much harder to analyze sociologically.

\begin{figure}[ht]
    \centering
    \includegraphics[width = \textwidth]{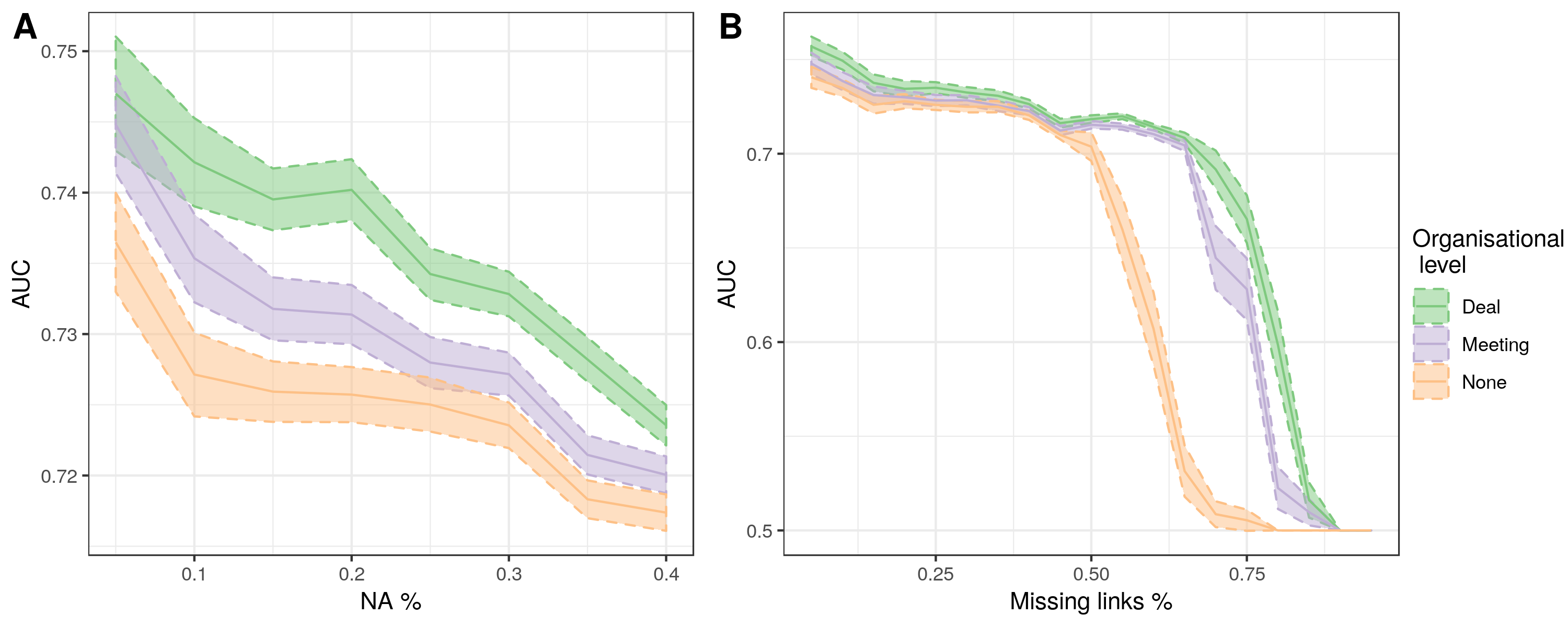}
    \caption{AUC of the prediction for  \textbf{A}: missing dyads, \textbf{B}: missing links, in function of the missing proportion for the inter-individual level. Colors represent different network at the inter-organizational level. None (beige) is equivalent to a single layer SBM on the individuals. The confidence interval is given by $mean \pm std error$.}
    \label{fig:predict_orga}
\end{figure}

\begin{remark}
  \OA{}{Another way to simulate missing data is to consider actor non-response like in  \citep{vznidarvsivc2012non}. In our case, it  corresponds to selecting a portion of the individuals at random and putting all their out-going dyads  to $\texttt{NA}$ (i.e. $\XR{ii'} = \texttt{NA}$ for all $i'$ if individual $i$ did not respond). Then we look at the stability of the clustering as in \cite{vznidarvsivc2012non, vznidarvsivc2019treating} (the ARI between the clustering of the individuals with the full data and the one with the missing data). By doing so, we notice in simulations (not reported here) that the clustering of the individuals is more stable when considering an MLVSBM on $(\XR{},X^{O,\text{deal}})$
  than when considering a unilevel SBM on $\XR{}$.
  This is one more clue in favor of the dependence between the two levels. }  
  \end{remark}

\begin{remark}  \OA{}{
\citet{vziberna2019blockmodeling} and \citet{vziberna2020k} also deals with this dataset from \citet{brailly2016embeddedness}. However, \cite{vziberna2019blockmodeling} uses the dataset collected in 2012 and 
\citet{vziberna2020k} gathers the datasets collected in 2011 and 2012 while we only use the 2011 dataset.
Moreover, different choices were made on the individuals and organizations to include  or not.
Thus, a direct comparison does not make sense. 
Applying \v{Z}iberna's method on the dataset we consider provides us with clusterings that somewhat agree on both levels (ARIs>0.6). We have checked that the difference derives from the fact that the two methods do not seek the same patterns.}
\end{remark}

\begin{sidewaysfigure}
    \centering
    \includegraphics[width = \textwidth]{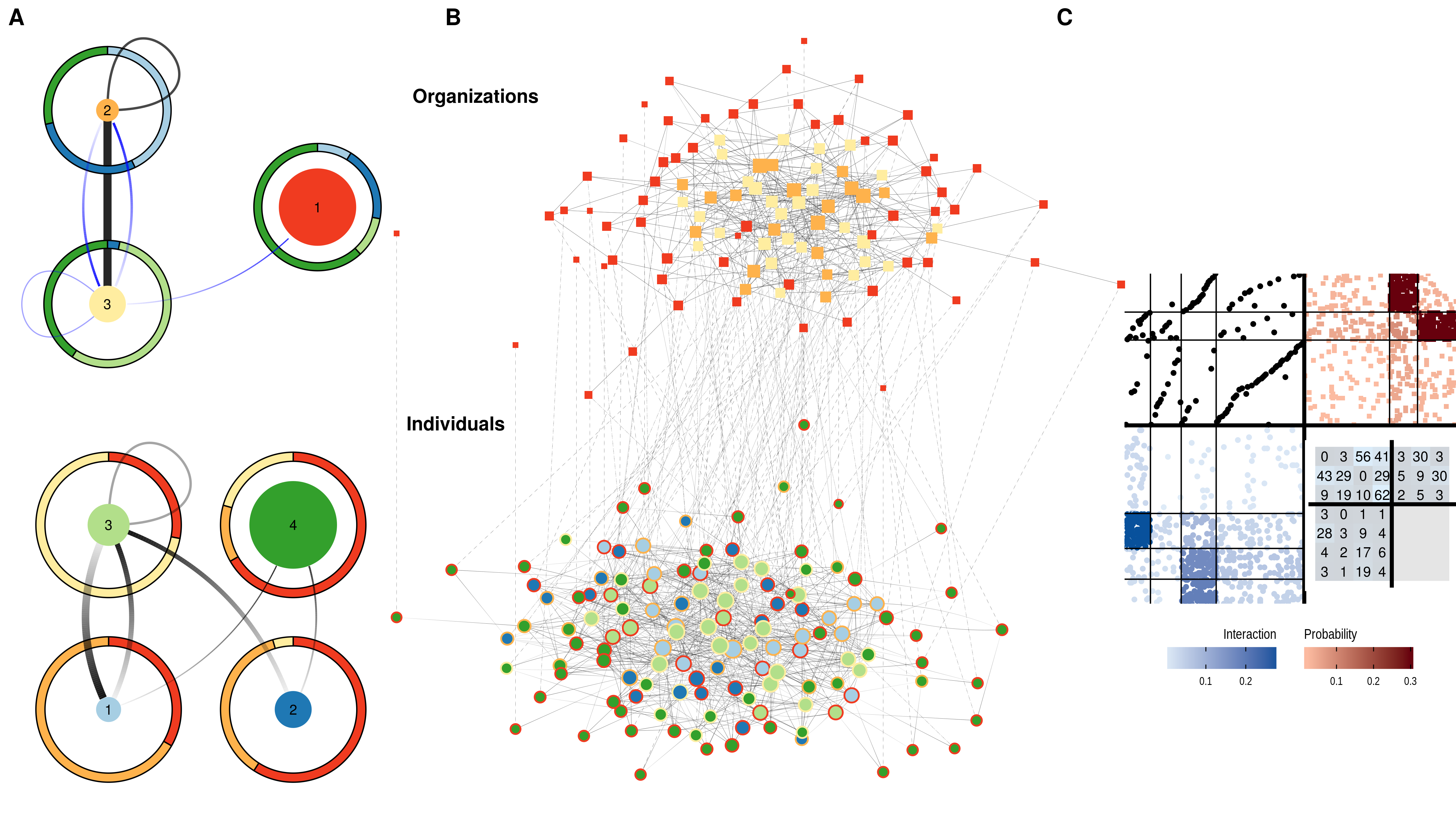}
    \caption{\OA{}{Multilevel network of the Promoshow East trade fair 2011. Above: the deal network for the organizations and below: the advice network for the individuals. 
    \textbf{A}: Mesoscopic view of the multilevel network. Nodes stand for the blocks, donut charts show the relation between $\ZL{}$ and $\ZR{}$. Black edges are the
    probabilities of connection $\alphar{}$ and $\alphal{}$, blue edges stand for  $\mathbb{P}(\XR{ii'} = 1 | \ZL{\aff_i} ,\ZL{\aff_{i'}})$, i.e. the probability of interaction between organizations through their individuals. For sake of clarity only edges with probabilities above the density are shown.
    \textbf{B}: View of the network. The size of a node is proportional to its in-degree. Colors represent the clustering obtained with the MLVSBM.
    \textbf{C}: Matrix representation of the multilevel network. At the bottom-left, the adjacency matrix of the advice network between individuals, at the top-right, the deal network between organizations, at the top-left, the affiliation matrix of the individuals to the organizations. Entries are reordered by block from left to right and bottom to top. Blocks are separated by thin lines and levels by thick lines. The entries of the bottom-right matrix are the parameters $\alphar{}$, $\gamma$ and $\alphal{}$ multiplied by $100$.}}
    \label{fig:deal_network}
\end{sidewaysfigure}

\subsection{Analysis and comments}
For the analysis, we use the  MLVSBM inferred from  the deal network.  We select $\widehat{\QL} = 3$  and  $\widehat{\QR} = 4$ blocks  and  the ICL is in favor of a dependence between the two levels.  
This network is plotted in Figure \ref{fig:deal_network} B and  we reordered the adjacency matrices of both levels by blocks in Figure \ref{fig:deal_network} C. 
 In Figure \ref{fig:deal_network}  A, we plot a synthetic view of the blocks of this multilevel network. The size of each node is proportional to the cardinal of each block. For the inter-organizational level, we link blocks of organizations by $\alphal{}$ (plain black edges) and by the probability of interactions of their individuals $\mathbb{P}(\XR{ii'} = 1 | \ZL{\aff_i} ,\ZL{\aff_{i'}})$ (gradual blue edges). The donut charts around the nodes is the parameter $\gamma$. For the inter-individual level, blocks of individuals are linked by $\alphar{}$ and the donut chart for a given block is the apportionment of each block of organizations in the individuals' affiliation. \\
 
We can now interpret the block with respect to the actors' covariates shown in Table \ref{tab:contingency}. At the inter-organizational level, block 1 (in red) is a residual group composed of $61$ organizations that are weakly connected to the rest of the organizations. Block 2 (in orange) consists of customers: broadcasters that come to the trade fair to buy programs and independent buyers who buy programs, planning to sell them later to broadcasters. We observe a non-null intra-block connection, but deals are mainly done between organizations of the blocks 2 and 3 (block 3 in yellow), the latter  mostly containing  distributors.

 At the inter-individual level,  blocks $1$ and $2$ consist of buyers (exclusively for block 1).  They differ in their affiliations, both are affiliated to   the second block of organizations but a larger proportion of the individuals of block 2 are affiliated to the residual block of organizations. They also differ  in the way they connect to blocks 3 and 4. Block 4 is a residual group consisting of roughly half of the individuals. It does not exhibit any particular pattern in its affiliations and is weakly connected, mainly inward connection from block 2. Block 3 consists of sellers  giving advices to individuals of block 2 and has reciprocal relationship with individuals of block 1. They are mainly affiliated to producing and distributing companies of block 3 of organizations. It is also the block that has the strongest intra-block connections. 

 The blue edges in Figure \ref{fig:deal_network}  A show that the organizations of blocks 2 and 3 and their respective individuals follow the same pattern for their inter-block connections but differ in their intra-block connections. Individuals affiliated to organizations of block 3 have above average intra-block connections while few contracts are signed between their organizations (mainly distributors).\\

\begin{table}[hbt!]
    \centering
    \begin{tabular}{c}
    \begin{tabular}{ccccccc}
        \toprule
        \multicolumn{2}{c}{Organizations} &
        \multicolumn{5}{c}{Covariates}\\
        \cmidrule(lr){1-2}\cmidrule(lr){3-7}
        Block & Size & Producer & Distributor & \parbox[t]{1cm}{Media\\ group} & \parbox[t]{1.9cm}{Independent buyer} & Broadcaster \\
        \midrule
        1 & 61 & 14 & 16 & 9 & 14 & 8 \\
        2 & 20 & 1  & 0  & 2 & 7  & 10\\
        3 & 28 & 3  & 19 & 5 & 1  & 0 \\
        \bottomrule
    \end{tabular}\\
     \newline
     \vspace*{1 cm}
     \newline
    \begin{tabular}{ccccccc}
        & & & & & & \\
        \toprule
        \multicolumn{2}{l}{Individuals} & \multicolumn{2}{c}{Covariates} & \multicolumn{3}{c}{Affiliation}\\
        \cmidrule(lr){1-2}\cmidrule(lr){3-4}\cmidrule(lr){5-7}
        Block & Size & Buyer & Seller & 1 & 2 & 3\\
        \midrule
        1 & 18 & 18 & 0  & 6  & 12 & 0 \\
        2 & 22 & 16 & 6  & 13 & 8  & 1 \\
        3 & 25 & 2  & 23 & 7  & 0  & 18\\
        4 & 63 & 15 & 48 & 42 & 8  & 13\\
        \bottomrule
    \end{tabular}
    \end{tabular}
        \caption{Contingency table of covariates
    and clustering for the organizations (top) and the individuals (bottom)}
    \label{tab:contingency}
\end{table}

These results confirm neo-structural insights into the functioning of markets. Competition between producers/distributors is strong: they all need to find broadcasting companies and distributors on the buying side. However, most of them arrive to the trade fair without updated information about the products in which buyers are interested in that year, their available budgets for each category of product, their willingness to negotiate, etc. The value of multilevel network analysis that is used here is to show that inter-individual personal relationships between individuals affiliated with competing organizations help manage the tensions between these directly competing organizations \citep{lazega2016effects, lazega2009theory}. This is where personal ties between individuals affiliated in these companies  -- especially among sellers and buyers, but also less visibly among sellers --  are important: they help manage the strong tensions between companies by creating \textit{coopetition}, i.e. cooperation among their competing firms. Here,  social/advice ties between buyers (blocks 1 and 2 of individuals) affiliated to buying companies in block 2 of organizations (broadcasting companies and distributors) exchange advice from sellers of block 3  representing production and distribution companies: this is the normal, stabilized, overlapping, commercial ties between companies embedded in social ties between representatives. 

 As seen above, block 3 has strong intra-block connections which may signal discreet coordination efforts between sellers as shown by \cite{brailly2016dynamics, brailly2016embeddedness}.  When a seller has closed a deal with a buyer, he/she can advise and update another seller – i.e. a coopetitor in terms of affiliation to a competing company – about other products in which this buyer is interested, what budget is left in his/her pocket, i.e. precious information for the next sellers. This kind of personal service is expected to be reciprocated over the years; otherwise the relationship decays.  This is the most unexpected phenomenon from an orthodox economic perspective and should lead to new perspectives in neo-structural economic sociology \citep{lazega2002interdependent}.

 This cross-level interdependence between inter-organizational ties and inter-individual ties is strong enough for companies to be unable to lay off its sales representatives. Having long tried to replace costly trade fairs with online websites and catalogues, companies realized that they still need the service that real persons and their personal relational capital provide in terms of multilevel management of coopetition \citep{lazega2020bureaucracy}.

\section{Discussion}\label{section:discussion}

In this paper, we propose an SBM for multilevel networks. We develop variational methods for the inference of the model and a criterion  that allows us to choose the number of blocks and to state on the independence between the levels at the same time. There are clear advantages at considering a joint modeling of the two levels over an independent model for  each level. Indeed, we show on some simulation studies that when we detect dependence between levels, it helps us to recover the block structure of a level with low signal thanks to the structure of the other level and also to improve the prediction of missing links or dyads. 
On the trade fair dataset, this joint modeling brought us a synthetic representation of the two networks unraveling their intertwined structure and provide new insights on the social organization. 

In lieu of a Bernoulli distribution, the edge distribution of any level may be extended to a valued distribution and/or to include edge covariates in a similar way as for the SBM \citep{mariadassou2010uncovering}. One way to account for the degree distribution would be to use nodes degrees as covariates, another would be to rewrite the edge distribution as the Degree Corrected SBM \citep{karrer2011stochastic}. 
 Our choice to model the interaction levels given the affiliations  ($\aff$ being fixed) is driven by the fact that, in a lot of applications, these affiliations are known and the object of the analysis is the interactions. We choose to consider a unique affiliation per individual  since this was the case on the datasets available to us, but this approach could be extended to a less restricted number of affiliations (this model is implemented in our \textsc{R} package).  We could even consider any hierarchical structure such as multi-scale networks to model the levels given the hierarchy or more generally multilayer networks by modeling the layers given the inter-layers.

Furthermore, our model is able to decide about the  independence of the structure of connections of the two levels. This is done by a model selection criterion. It would be interesting to test (in a statistical meaning)  this independence but we know that the  variance of our estimators is underestimated because of  
the variational approach (see \cite{blei2017variational} for a review). Besides, sociological studies stated that some individuals benefit  more than others from their organization's interactions \citep{lazega2015multilevel}, which could lead us to consider more local independence between levels. 

 For multiplex networks, \cite{de2017community} use dyad predictions as a way to define interdependence between layers while  \cite{stanley2016clustering} make a clustering of layer by aggregating the most similar. Our work considers multilevel networks where each level has nodes of different natures and Figure \ref{fig:predict_orga} shows that the dependence between levels leads to  a better recovery of missing information. This can be used to help data collection or to correct spurious information on existing data as suggested in \cite{clauset2008hierarchical} or \cite{guimera2009missing}. Indeed, one might imagine  that the data of one level  may  be easier to collect or to verify than the other one (for instance because it is public,  already exists or is   cheaper to collect).   \OA{}{Thus, we think that  this approach could be used to leverage the interdependence in a multilevel network in order  to compensate for some missing or spurious information on a given level which is known to be difficult to observe.}

\section*{Acknowledgements}
The authors would like to thank Julien Brailly for providing the dataset. 
This work was supported by a public grant as part of the Investissement d'avenir project, reference ANR-11-LABX-0056-LMH, LabEx LMH. This work was partially supported by the grant ANR-18-CE02-0010-01 of the French National Research Agency ANR (project EcoNet). This project received financial support from INRAE and CIRAD as part of the SEARS project funded by the GloFoods metaprogram. This work was presented and discussed within the framework of working days organized by the MIRES group (with the financial support of INRAE) and the GDR RESODIV (with the financial support of CNRS).

\bibliographystyle{chicago}
\bibliography{scholar}

\appendix

\renewcommand{\theequation}{\thesection.\arabic{equation}}

\section{Proof of Proposition \ref{prop:independence}} \label{appendix:independence}
\setcounter{proposition}{0}

\begin{proposition}
In the MLVSBM, the two following properties are equivalent: 
\begin{enumerate}
  \item $\ZR{}$ is independent on  $\ZL{}$, 
  \item $\gamma_{kl} = \gamma_{kl'} \quad \forall l, l' \in \{1, \dots, \QL\}$,
\end{enumerate}
and imply that:  
\begin{enumerate}
  \item[3.] $\XR{}$ and $\XL{}$ are independent.
\end{enumerate}
\end{proposition}

\begin{proof}
We first derive an expression for $\ell_{\gamma}(\ZR{}) =  \ell_{\gamma}(\ZR{}|\aff)$:
\begin{align*}
    \ell_{\gamma}(\ZR{}|\aff) & = \int_{\ZL{}}\ell_{\gamma}(\ZR{} | \aff, \ZL{}) d\mathbb{P}(\ZL{})\\
    &=\sum_{l_1, \dots, l_{\nbl}}\ell_{\gamma}(\ZR{} | \aff, \ZL{1} = l_1, \dots, \ZL{\nbl} = l_{\nbl}) \mathbb{P}(\ZL{1} = l_1, \dots, \ZL{\nbl} = l_{\nbl}) \\ 
    & = \sum_{l_1, \dots, l_{\nbl}}\prod_{j}\left(\prod_{i} \ell_{\gamma}(\ZR{i} | \aff, \ZL{\aff_i} = l_{\aff_i})\right) \mathbb{P}(\ZL{j} = l_j) \\
    & = \sum_{l_1, \dots, l_{\nbl}}\prod_{j}\left(\prod_{i, k} \gamma_{kl_j}^{\mathds{1}_{\ZR{i} = k}\aff_{ij}}\right) \pil{l_j} = \prod_{j} \sum_{l} \prod_{i,k}\gamma_{kl}^{\aff_{ij}\mathds{1}_{\ZR{i} = k}} \pil{l} 
\end{align*}
where $A_i = \{j : \aff_{ij} = 1\}$. 

$2. \Rightarrow 1.$: Assume that  $\gamma_{kl} = \gamma_{kl'} \quad \forall l, l' \in \{1, \dots, \QL\}$, then:
\begin{align*}
    \ell_{\gamma{}}(\ZR{} | \ZL{}, \aff) & = \prod_{k, l} \gamma_{kl}^{\sum_{i,j} \aff_{ij} \mathds{1}_{\ZR{i} = k} \mathds{1}_{\ZL{j} = l}}  = \prod_{k} \gamma_{k1}^{\sum_{i,j} \aff_{ij} \mathds{1}_{\ZR{i} = k} \sum_l \mathds{1}_{\ZL{j} = l}} \\& = \prod_{i, k} \gamma_{k1}^{ \mathds{1}_{\ZR{i} = k}}, 
\end{align*}
and 
\begin{align*}
    \ell_{\gamma}(\ZR{}|\aff) & =  \prod_{j} \sum_{l} \prod_{i,k}\gamma_{kl}^{\aff_{ij}\mathds{1}_{\ZR{i} = k}} \pil{l} \\
    & = \prod_{j} \prod_{i,k}\gamma_{k1}^{\aff_{ij}\mathds{1}_{\ZR{i} = k}} \sum_{l}\pil{l}  =  \prod_{i,k}\gamma_{k1}^{\mathds{1}_{\ZR{i} = k}},
\end{align*}
hence $\ell_{\gamma{}}(\ZR{} | \ZL{}, \aff) = \ell_{\gamma{}}(\ZR{}|\aff)$.

$1. \Rightarrow 2.$: 
Assume that $\ell_{\gamma{}}(\ZR{} | \ZL{}, \aff) = \ell_{\gamma{}}(\ZR{}|\aff)$ for any values of $\ZR{}, \ZL{}$, then 
in particular
$\ell_{\gamma{}}(\ZR{1} | \ZL{}, \aff) = \ell_{\gamma{}}(\ZR{1}|\aff)$. Assuming that individual $1$ belongs to organization $j$, we can write, for any $k$: 
$$\pr{\ZR{1} = k | \ZL{j},A_{ij}=1} = \gamma_{k\ZL{j}}.$$
However, this quantity does not depend on $\ZL{j}$ so $ \gamma_{k\ZL{j}} = \gamma_{k \cdot}$  for any value of $k$ and $\ZL{j}$. And so we have   $ \gamma_{k \ell} =  \gamma_{k\ell'}$ for any $(\ell,\ell')$. \\

$1. \Rightarrow 3.$: 
\begin{align*}
    \ell_{\alphar{},\alphal{}}&(\XR{}, \XL{} | \aff)  =
    \int_{z^I, z^O} \ell_{\alphar{},\alphal{}}(\XR{}, \XL{} | \aff, \ZR{} = z^I, \ZL{} = z^O) \mathbb{P}(\ZR{} = z^I, \ZL{} = z^O) \text{d}z^I\text{d}z^O\\
    & = \int_{z^I, z^O} \ell_{\alphar{}}(\XR{}| \ZR{} = z^I) \mathbb{P}(\ZR{} = z^I | A,  \ZL{} = z^O)  \ell_{\alphal{}}(\XL{}| \ZL{} = z^O)\mathbb{P}(\ZL{} = z^O)\text{d}z^I\text{d}z^O\\
    & = \int_{z^I} \ell_{\alphar{}}(\XR{}| \ZR{} = z^I) \mathbb{P}(\ZR{} = z^I)\text{d}z^I \int_{z^O} \ell_{\alphal{}}(\XL{}| \ZL{} = z^O)\mathbb{P}(\ZL{} = z^O)\text{d}z^O\\
    & = \ell_{\alphar{}}(\XR{})\ell_{\alphal{}}(\XL{})
\end{align*}
which is the definition of  the independence.  
\end{proof}

\section{Proof of Proposition \ref{prop:identif}}\label{appendix:identifiability}

\begin{proposition}
The stochastic block model for multilevel networks is identifiable up to label switching under the following assumptions:
\begin{enumerate}
\item[$\mathcal{A}$1.] All coefficients of  $\alphar{}\cdot\gamma\cdot\pil{}$ are distinct and all coefficients of $\alphal{}\cdot\pil{}$ are distinct.
\item[$\mathcal{A}$2.] $\nbr \geq 2\QR$ and $\nbl \geq \max(2\QL, \QL+\QR-1)$.
\item[$\mathcal{A}$3.] At least $2\QR$ organizations contain one individual or more.
\end{enumerate}
\end{proposition}

\begin{proof}

Let $\theta = \{ \pil{}, \gamma, \alphar{}, \alphal{} \}$ be the set of parameters and $\mathbb{P}_{X}$ the distribution of the observed data. We will prove that there is a unique $\theta$ corresponding to $\mathbb{P}_{X}$. More precisely, in what follows, we will compute the probabilities of some particular events, from which we will derive a unique expression for the unknown parameters. 
The beginning of the proof --identifiability of $\pil{}$ and $\alphal{}$-- is mimicking the  one given in \cite{celisse2012consistency}. The last steps of the proof are original work. 

\paragraph{Notations. } For the sake of simplicity, in what follows,  we use the following shorten notation: 
$$ x_{i:k} := (x_i, \dots,x_k),  \quad X_{j,i:k} = (X_{ji}, \dots, X_{jk})\,.$$
Moreover, $\{X_{j,i:k} = 1\}$  stands for $ \{X_{ji} = 1, \dots, X_{jk}=1\}$. 

\paragraph{Identifiability of $\pil{}$}
For any $l=1, \dots, \QL$, let $\tau_l$ be the following probability: 
\begin{equation}\label{eq:identif1}
    \tau_l = \pr{\XL{ij} = 1 | \ZL{i} = l} =  \sum_{l'} \alphal{ll'}\pil{l'} = (\alphal{} \cdot \pil{})_{l}, \quad  \forall (i,j). 
\end{equation}
Moreover, a quick computation proves that
\begin{equation}\label{eq:identif1bis}
\pr{\XL{i,j : (j+k)}  = 1 | \ZL{i} = l } =  \tau_l^{k+1}
\end{equation}
According to Assumption $\mathcal{A}$1, the coordinates of vector $(\tau_1, \dots, \tau_{\QL})$  are all different. Hence, the Vandermonde  matrix $R^O$ of size $\QL \times \QL$ such that 
\[ R^O_{il} = (\tau_l)^{i-1}, \quad 1 \leq i \leq \QL, \quad 1 \leq l \leq \QL \]
is   invertible. 
We define  $u^O_i$ as follows: 
\begin{equation*}
\begin{array}{ccll}
 u^O_i &=&  \mathbb{P}_{\obs,\theta}(\XL{1,2:(i+1)} = 1)  &  \mbox{ for }1 \leq i \leq 2\QL - 1 \\
 u^O_0 &=& 1.
 \end{array}
\end{equation*}
The existence of $(u^O_i)_{i=0, \dots, 2\QL - 1}$ comes from Assumption $\mathcal{A}$2 ($\nbl \geq 2 \QL$). Moreover, the $(u_i^O)_{i=0, \dots, 2\QL-1}$ are calculated from the marginal distribution  $\mathbb{P}_{X}$. We will use these quantities to identify the parameters $(\pil{}, \alphal{})$. 

\vspace{1em}

\noindent First we have, for  $1 \leq i \leq 2\QL - 1$: 
\begin{eqnarray*}\label{eq:ident1}
u^O_i &=& \sum_{l = 1}^{\QL}\pr{\XL{1,2:(i+1)} = 1 | \ZL{1} = l } \pr{\ZL{1} = l} = 
 \sum_{l = 1}^{\QL}  \tau_l^i  \pil{l}, 
\end{eqnarray*}
using Equation \eqref{eq:identif1bis}. 
Now, let us define $M^O$ a $(\QL +1) \times \QL$ matrix such that:
\begin{equation}\label{eq:ident2}
 M_{ij}^O = u^O_{i+j-2} = \sum_{l=1}^{\QL} \tau_{l}^{i-1} \pil{l} \tau_{l}^{j-1}, \quad 1 \leq i \leq \QL +1, \quad 1 \leq j \leq \QL.
\end{equation}
For $k \in \{1,\dots,\QL+1\}$, we define  $\delta_{k}$ as $\delta_k = \text{Det}(M^O_{-k})$ where    $M_{-k}^O$ is the square matrix corresponding to $M^O$ without the $k$-th row. 
Let  $B^O$ be the polynomial function defined as: \begin{equation}\label{eq:identif BO}
B^O(x) = \sum_{k = 0}^{\QL} (-1)^{k+\QL}\delta_{k+1} x^k. 
\end{equation}
\begin{itemize}
\item $B^O$ is of degree $\QL$. Indeed, $  \delta_{\QL +1 } = \mbox{det}(M_{-(\QL+1)}^O)$ 
and $M_{-(\QL+1) }= R^OD_{\pil{}}{R^O}'$ where $D_{\pil{}} = \text{diag}(\pil{})$. As a consequence, 
$M_{-(\QL+1)}^O$ is the product of invertible matrices then $\delta_{\QL+1} \neq 0$ and  we can conclude. 
 \item Moreover, $\forall l= 1, \dots, \QL$, $B^O(\tau_l) =0$ . Indeed,  $B^O(\tau_l) = \mbox{det}(N_l^O)$ where  $N_l^O$ is  the concatenated matrix   $N_l^O = \left(M^O\, |\, V_{l} \right)$ with  $V_{l} = [1, \tau_l, \dots, \tau_l^{\QL}]'$ (computation of the determinant development against the last column).  However, from Equation \eqref{eq:ident2}, we have $M^O_{\bullet j} = \sum_l \tau_{l}^{j-1} \pil{l} V_l$,  i.e. each column vector of $M^O$  is a linear combination of $V_1, \dots, V_{\QL}$. As a consequence,  $\forall l= 1, \dots, \QL$,  $N_l^O$ is of rank $< \QL +1$,  and so $B^O(\tau_l)=0$. 
\end{itemize}
The $(\tau_l)_{ l= 1, \dots, \QL}$ being the  roots of $B$,  they can be expressed in a unique way (up to   label switching)  as functions of $(\delta_k)_{k=0,\dots, \QL}$, which themselves are derived from  $\mathbb{P}_{\obs,\theta}$.  As a consequence,  the identifiability of $R^O$ is derived from the identifiability of $(\tau_l)_{ l= 1, \dots, \QL}$. Using the fact that  $ D_{\pil{}} = {R^O}^{-1}M_{-\QL}^{O}{R^O}^{'-1}$, we can identify $\pil{}$ in a unique way.

\paragraph{Identifiability of $\alphal{}$}
 For $1 \leq i, j \leq \QL$,  we define $U_{ij}$ as follows: 
$$
U^O_{ij} = \pr{ \XL{1,2:(i+1)}  = 1,  \XL{2,\, (\nbl - j + 2):\nbl} =1 }$$
with 
$ U^O_{i1} = \pr{ \XL{1,2:(i+1)} =1 }$.
\begin{align*}
U^O_{i, j}  = \sum_{l_1, l_2} \tau_{l_1}^{i-1} \pil{l_1} \alphal{l_1l_2} \pil{l_2} (\tau_{l_2})^{j-1}, \quad  \forall 1 \leq i, j \leq \QL, 
\end{align*}
 and as consequence
 $U^O = R^O D_{\pil{}} \alphal{} D_{\pil{}} {R^O}'.$  
$D_{\pil{}}$  and $R^O$ being  invertible, we get: 
$ \alphal{} = D_{\pil{}}^{-1} {R^O}^{-1} U^O {R^O}^{'-1} D_{\pil{}}^{-1}$. And so $U_O$ is uniquely derived from  $\mathbb{P}_{X}$, so $\alphal{}$ is identified. 

\paragraph{Identifiability of $\alphar{}$}
To identify $\alphar{}$, we have to take into account the affiliation matrix $\aff$. 
Without  loss of generality, we reorder the entries of both levels such that the affiliation    matrix  $\aff$ has its $2\QR \times 2\QR$ top left  block being an identity matrix  (Assumption $\mathcal{A}3$). 
\\ 

\begin{itemize}
\item  For any $k  = 1 \dots, \QR$ and for $i = 2, \dots, 2\QR$, let $\sigma_{k}$ be the    probability $ \pr{\XR{1i} = 1 | \ZR{1} = k, A}$, $A$ being such that $A_{jj} = 1, \forall  j= 1, \dots, 2\QR$. 
\begin{align*}
\sigma_{k} &= \pr{\XR{1i} = 1 | \ZR{1} = k, A} \\
& = \sum_{k'} \pr{\XR{1i} = 1 | \ZR{1} = k, \ZR{i} = k'} \pr{\ZR{i} = k' |  \ZR{1} = k,A}\,.\\
\end{align*}
Moreover, 
\begin{eqnarray}\label{eq:identif10}
\pr{\ZR{i} = k' |  \ZR{1} = k,A} &=&  \sum_{l}  \pr{\ZR{i} = k' |  \ZL{i}  = l,  \ZR{1} = k, A}\pr{\ZL{i} =l |  \ZR{1} = k,A}\nonumber \\
&=&   \sum_{l} \gamma_{kl} \pr{\ZL{i} =l |  \ZR{1} = k,A}\,.
\end{eqnarray}
However, by Bayes' formula
$$ \pr{\ZL{i} =l |  \ZR{1} = k,A}  = \frac{\pr{ \ZR{1} = k | \ZL{i} =l , A}\pr{\ZL{i} =l}}{\pr{\ZR{1} = k,A}}\,.$$
Taking into the fact that $i \neq 1$ and $A$ is such that $1$ belongs to  organization  $1$ and $i$ to organization $i$, we have:  $\pr{ \ZR{1} = k | \ZL{i} =l , A} = \pr{ \ZR{1} = k | A }$.  And so $$\pr{\ZL{i} =l |  \ZR{1} = k,A} = \pr{ \ZL{i} =l| A }  = \pil{l}.$$ 
Consequently, from Equation \eqref{eq:identif10}, we have: 
$$\pr{\ZR{i} = k' |  \ZR{1} = k,A} =  \sum_{l}  \gamma_{k' l} \pil{k}$$ 
and so: 
\begin{align*}
\sigma_{k}  &=  \sum_{k'}   \pr{\XR{1i} = 1 | \ZR{1} = k, \ZR{i} = k'}  \sum_{l}  \gamma_{k' l} \pil{k} \\
& = \sum_{k'l} \alphar{kk'}\gamma_{k'l}\pil{l} = (\alphar{} \cdot \gamma \cdot \pil{})_{k} \\
& = (\alphar{}  \cdot \pir{})_{k}, \quad \quad  \mbox{ where }  \pir{} = \gamma\cdot\pil{}. 
\end{align*}

\item 
Now, we prove that $\forall i = 1, \dots, 2\QR -1$, 
  \begin{equation} \label{eq:identif3} \pr{\XR{1,2: (i+1)} =  1 |\ZR{1} = k,A} =\sigma_k ^i.\end{equation}  
Indeed, 
 \begin{align*}
&\pr{\XR{1,2:(i+1)} = 1 |\ZR{1} = k,A} \\
 =&\sum_{k_{2:(i+1)}}\pr{\XR{1,2:(i+1)} = 1 |\ZR{1:(i+1)} = (k,k_{2:(i+1)}),\ZR{1} = k }  \pr{\ZR{2:(i+1)}=k_{2:i+1}|\ZR{1} = k,A}\\
 =&\sum_{k_{2:(i+1)}}\pr{\XR{1,2:(i+1)} = 1 |\ZR{1:(i+1)} = (k,k_{2:(i+1)}) }  \pr{\ZR{2:(i+1)}=k_{2:i+1} | A}\\
=&\sum_{k_{2:(i+1)}}\pr{\XR{1,2:(i+1)} = 1 |\ZR{1:(i+1)} = (k,k_{2:(i+1)}) }  \sum_{l_{2:(i+1)}} \pr{ \ZR{2:(i+1)} = k_{2:(i+1)}, \ZL{2:(i+1)}=l_{2:(i+1)} ,A}.   
\end{align*} 
Note that, to go from line $2$ to line $3$, we used the fact that $  \pr{\ZR{2:(i+1)}=k_{2:i+1}|\ZR{1} = k,A} =  \pr{\ZR{2:(i+1)}=k_{2:i+1} | A}$, which is due the the particular structure of $A$ (left diagonal block of size at least $2\QR$, i.e. for any  $ i'= 1,\dots , 2 \QR$,    individual $i'$ belongs to organization $i'$). Moreover,  we can write:  
\begin{eqnarray*}
 &&  \pr{ \ZR{2:(i+1)} = k_{2:(i+1)}, \ZL{2:(i+1)}=l_{2:i+1} |A} \\
&=&  \left[ \prod_{\lambda=2,\ldots i+1}   \pr{\ZR{\lambda} =k_\lambda |\ZL{\lambda}  = l_{\lambda}}  \pr{\ZL{\lambda} = l_{\lambda}} \right] \\
&=&   \left[ \prod_{\lambda=2,\ldots i+1}    \gamma_{k_\lambda l_\lambda}   \pil{\lambda}\right] \,. \\
\end{eqnarray*}

Moreover, by conditional independence of the entries of the matrix $\XR{}$ given the clustering we have: 
$$\pr{\XR{1,2:(i+1)} = 1 |\ZR{1} = k,\ZR{2:(i+1)}=k_{2:(i+1)}   } = \prod_{\lambda=2,\ldots i+1}  \alphar{kk_\lambda}.$$
As a consequence, 
\begin{align*}
\pr{\XR{1,2:(i+1)} = 1 |\ZR{1} = k,A} &=  \sum_{k_{2:(i+1)},l_{2:(i+1)}} \prod_{\lambda=2,\ldots i+1} \alphar{kk_\lambda} \gamma_{k_\lambda l_\lambda} \pil{\lambda}\\ 
 &=\prod_{\lambda=2,\ldots i+1} \sum_{k_\lambda,l_\lambda} \alphar{kk_\lambda} \gamma_{k_\lambda l_\lambda} \pil{\lambda}\nonumber  = \sigma_k^{i}
 \end{align*}
 
\item Then we define $(u_i^I)_{i=0, \dots, 2 \QR - 1}$, such that  $u^I_0 =1$  and   $ \forall 1\leq i \leq  2 \QR - 1$: 
\begin{align*}
u^I_i & = \pr{\XR{1,2:(i+1)} = 1  | \aff} \\
    & = \sum_{k,   l} \pr{\XR{1,2:(i+1)} = 1   |\ZR{1} = k}  \pr{\ZR{1} = k | \ZL{1} = l, A} \pr{\ZL{1} = l} \\
    & = \sum_{k} \sigma_k^{i} \underbrace{ \sum_{l}  \gamma_{kl}\pil{l}}_{ =\pir{k}} \\
    & = \sum_{k } \sigma_k^{i}  \pir{k}.   
\end{align*}
Note that the $(u^I)$'s can be defined because $\nbr \geq 2 \QR$ (assumption $\mathcal{A}2$). 

 \item  To conclude we use the same arguments as the ones used for the identifiability of $\alphal{}$, i.e.   we  define $M^I$ a  $(\QR +1)  \times \QR$ matrix such that $M_{ij}^I = u^I_{i+j-2}$ together with the matrices $M^I_{-k}$ and  the polynomial function $B^I$ (see Equation \eqref{eq:identif BO}). 
Let $R^I$ be a $\QR \times \QR$ matrix such that $R^I_{ik} = \sigma_k^{i-1}$. $R^I$ is an invertible Vandermonde matrix because of assumption $\mathcal{A}1$ on $\alphar{}\cdot  \gamma \cdot \pil{}$.  As before, $R^I$ can be identified in unique way from $B^I$. 
Then, noting that $M_{-(\QR+1)}^{I}  = R^I D_{\pir{}} {R^I}^{'}$ where $D_{\pir{}} = \text{diag}(\pir{}) = \text{diag} (\gamma \cdot \pil{})$, we obtain: 
$D_{\pir{}} = ({R^I})^{-1}M_{-\QR}^{I}({R^I}^{'-1}),$
which is uniquely defined by $\mathbb{P}_X$. 
Now, let us introduce 
$$U^I_{ij} =\pr{\XR{1,2:(i+1)} = 1,   \XR{2,\, (\nbr - j + 2):\nbr} = 1}$$
with 
$ U^I_{i1} = \pr{\XR{1,2:(i+1)}=1}$.
Then we have  $U^I = R^I D_{\pir{}} \alphar{} D_{\pir{}} {R^I}'$ and so \\
$ \alphar{} = D_{\pir{}}^{-1} (R^I)^{-1} U^I (R^I)'^{-1} D_{\pir{}}^{-1}$. 
As a consequence, $\alphar{}$ is uniquely identified from $\mathbb{P}_X$. 
\end{itemize}

\paragraph{Identifiability of $\gamma$}

For any $2 \leq i \leq \QR $  and $2 \leq j \leq \QL$, let $U^{IO}_{i,j}$  be the probability that  $\XR{1,2:i} = 1$ and $\XL{1,(i+1):(i+j-1)}=1$.  Note that the $U^{IO}_{i,j}$ can be defined because $\nbl \geq \QR + \QL -1$   and $\nbr \geq \QR$ (assumption $\mathcal{A}2$).
\begin{itemize}
\item Then, for all $2 \leq i \leq \QR $  and $2 \leq j \leq \QL$,
\begin{eqnarray}\label{eq:identif9}
U^{IO}_{ij} & = &  \pr{\XR{1,2:i} = 1, \XL{1,(i+1):(i+j-1)}=1 | \aff}\nonumber\\
& = &\sum_{k, l}\pr{\XR{1,2:i} =1, \XL{1,(i+1):(i+j-1)}=1 | \aff, \ZR{1} = k, \ZL{1} = l}\nonumber\\
&& \times\pr{\ZR{1} = k, \ZL{1} = l,A} \,.  
\end{eqnarray}

\item 
We first prove that : 
\begin{eqnarray}\label{eq:identif8}
  \pr{\XR{1,2:i} = 1, \XL{1,i+1 : i+j - 1} =1 | \aff, \ZR{1} = k, \ZL{1} = l}  = 
  \sigma_k^{i-1} \tau_{l}^{j-1}\,.
\end{eqnarray}
Indeed, 
 \begin{eqnarray}\label{eq:identif7}
 &&\mathbb{P}(\XR{1, 2:i} = 1,  \XL{1,(i+1) :( i+j-1)}= 1 | \aff, \ZR{1} = k, \ZL{1} = l)  = \nonumber \\
&& = \sum_{k_{2 : i},l_{2: \nbl}}\pr{\XR{1, 2:i} =1,   \XL{1,(i+1) :( i+j-1)}= 1 |\ZR{1: i} = (k,k_{2 : i}), \ZL{}  =(l,l_{2: \nbl}),A}\nonumber \\
&&  \quad  \quad \times \quad  \pr{\ZR{2: i} = k_{2 : i}, \ZL{2:\nbl}  =l_{2: \nbl}|\ZR{1} = k, \ZL{1} = l,A}\nonumber \\
&&=  \sum_{k_{2 : i},l_{2: \nbl}}\pr{\XR{1, 2:i} = 1 |\ZR{1: i} = (k,k_{2 : i})} \nonumber \\
&& \quad  \quad \times \quad \pr{ \XL{1,(i+1) :( i+j-1)}= 1 | \ZL{1}  =l, \ZL{(i+1) : (i+j-1)}  =l_{(i+1) : (i+j-1)}}\nonumber \\
&& \quad  \quad \times \quad \pr{\ZR{2: i} = k_{2 : i}, \ZL{2:\nbl}  =l_{2: \nbl}|\ZR{1} = k, \ZL{1} = l,A} \,.
\end{eqnarray}
Moreover,  let us have a look at $\pr{\ZR{2: i} = k_{2 : i}, \ZL{}  =l_{2: \nbl}|\ZR{1} = k, \ZL{1} = l,A} $: 
\begin{eqnarray*}
&&\pr{\ZR{2: i} = k_{2 : i}, \ZL{2:\nbl}=l_{2: \nbl}|\ZR{1} = k, \ZL{1} = l,A}   \\
&=&\pr{\ZR{2: i} = k_{2 : i} |  \ZL{2:\nbl}  = l_{2: \nbl},\ZR{1} = k, \ZL{1} = l,A}  \times   \pr{\ZL{2:\nbl}  =l_{2: \nbl}|\ZR{1} = k, \ZL{1} = l,A}\,.
\end{eqnarray*}
Because $A$ has a diagonal block of size $\geq  \QR$ , we have,  for any $i = 1 ,\dots, \QR $, $A_{ij} = 1$ if $j = i$, $0$ otherwise, we have
\begin{itemize}
\item[$\bullet$] 
$   \pr{\ZR{2: i} =  k_{2 : i} |  \ZL{2:\nbl}  = l_{2: \nbl},\ZR{1} = k, \ZL{1} = l,A}  =   \pr{\ZR{2: i} = k_{2 : i}  |  \ZL{2:i}  = l_{2: i} }, $ 
\item[$\bullet$]  
$ \pr{ \ZL{2:\nbl}  =l_{2: \nbl}|\ZR{1} = k, \ZL{1} = l,A}   =  \pr{\ZL{2:\nbl}  =l_{2: \nbl}}  \,.$  
\end{itemize}
As a consequence, 
\begin{eqnarray*}
&& \pr{\ZR{2: i} = k_{2 : i}, \ZL{2:\nbl}  =l_{2: \nbl} |\ZR{1} = k, \ZL{1} = l,A}  =\\
&& \quad  \pr{\ZR{2: i} = k_{2 : i}  |  \ZL{2:i}  = l_{2: i}} \pr{\ZL{2:i}  = l_{2: i}}       \pr{\ZL{(i+1) :( i+j-1)}  =l_{(i+1) :( i+j-1)}} \\
&&\times    \pr{\ZL{(i+j) : \nbl}  =l_{(i+j) : \nbl}}\,.
\end{eqnarray*}
Going back to Equation \eqref{eq:identif7} and decomposing the summation we obtain: 
\begin{eqnarray*}
 &&  \mathbb{P}(\XR{1, 2:i} = \XL{1,(i+1):(i+j-1)}=1 | \aff, \ZR{1} = k, \ZL{1} = l)  \\
&=&  \sum_{k_{2 : i},l_{2: \nbl}}\pr{\XR{1, 2:i} = 1 |\ZR{1: i} = (k,k_{2 : i})} \nonumber \\
&& \quad  \quad \times \quad \pr{ \XL{1,(i+1) :( i+j-1)}= 1 | \ZL{1}  =l, \ZL{(i+1) : (i+j-1)}  =l_{(i+1) : (i+j-1)}}\nonumber \\
&& \quad  \quad \times \quad   \pr{\ZR{2: i} = k_{2 : i}  |  \ZL{2:i}  = l_{2: i}}   \pr{\ZL{2:i}  = l_{2: i}} \pr{\ZL{(i+1) :( i+j-1)}  =l_{(i+1) :( i+j-1)}}   \\
&& \quad \quad  \times  \quad    \pr{\ZL{(i+j) : \nbl}  =l_{(i+j) : \nbl}} \\
&= &    \sum_{k_{2 : i}}\pr{\XR{1, 2:i} = 1 |\ZR{1: i} = (k,k_{2 : i})}\sum_{l_{2: i}}  \pr{\ZR{2: i} = k_{2 : i}  |  \ZL{2:i}  = l_{2: i} }   \pr{\ZL{2:i}  = l_{2: i}} \\
&&  \quad \quad \sum_{l_{(i+1): (i+j-1)}}  \pr{ \XL{1,(i+1) :( i+j -1 )}= 1 | \ZL{1}  =l, \ZL{(i+1) : (i+j-1)}  =l_{(i+1) : (i+j-1)})} \\
&& \quad  \quad  \quad  \quad  \quad \quad \times \underbrace{\pr{\ZL{(i+1) : (i+j-1)}  =l_{(i+1): (i+j-1)}}}_{=\pr{\ZL{(i+1) : (i+j-1)}  =l_{(i+1): (i+j-1)} |  \ZL{1}=l}} \underbrace{\sum_{l_{(i+j):\nbl}} \pr{\ZL{(i+j) : \nbl}  =l_{(i+j) : \nbl}}}_{=1}\\
&=&  \sum_{k_{2 : i}}\pr{\XR{1, 2:i} = 1 |\ZR{1} = k, \ZR{2: i} =k_{2 : i}} \pr{\ZR{2: i} = k_{2 : i} | A}  \times  \pr{ \XL{1,(i+1) :( i+j-1)}= 1 | \ZL{1}  =l}\\
&=&  \sum_{k_{2 : i}}\pr{\XR{1, 2:i} = 1 |\ZR{1} = k, \ZR{2: i} =k_{2 : i}} \pr{\ZR{2: i} = k_{2 : i} | \ZR{1} = k,A}\\
&&\times  \pr{ \XL{1,(i+1) :( i+j-1)}= 1 | \ZL{1}  =l}\\
&=& \pr{\XR{1, 2:i} = 1 |\ZR{1} = k,  A}  \pr{ \XL{1,(i+1) :( i+j-1)}= 1 | \ZL{1}  =l}\,.
\end{eqnarray*}
Finally, we have : 
\begin{eqnarray*}
 \pr{\XR{1, 2:i} = 1 |\ZR{1} = k,  A}&=& \sigma_k^{i-1},   \quad \mbox{from  Equation \eqref{eq:identif3} }  \\
 \pr{ \XL{1,(i+1) :( i+j-1)}= 1 | \ZL{1}  =l} &= &\tau_l ^{j-1},
\end{eqnarray*}
and so, we have proved equality \eqref{eq:identif8}. 
\item 
Now, $A_{11} = 1$ implies  $\pr{\ZR{1} = k, \ZL{1} = l|A} = \gamma_{kl}\pil{l} $  and combining this result with Equations   \eqref{eq:identif8} and  \eqref{eq:identif9}
leads to: 
$ U^{IO}_{ij} = \sum_{k, l}\sigma_{k}^{i-1}  \gamma_{k l} \pil{l}  \tau_{l}^{j-1}$. 
Setting 
\begin{eqnarray*}
U^{IO}_{1j} &=&   \pr{\XL{1,i+1}=1, \dots, \XL{1, i+j-1}=1 | \aff}  = \sum_{k, l}  \gamma_{k l} \pil{l}  \tau_{l}^{j-1}, \quad \mbox{for $j>1$} \\
U^{IO}_{i1} &=&   \pr{\XR{12} = \dots =  \XR{1,i} = 1 | \aff}  = \sum_{k, l}  \gamma_{k l} \pil{l},    \quad \mbox{for $i>1$}\\
U^{IO}_{11}&=& 1
\end{eqnarray*}
we obtain  the following  matrix expression for $U^{IO}$: 
$ U^{IO} = R^{I} \gamma D_{\pil{}} R^{O'}$ where $ U^{IO} $ is completely defined by $\mathbb{P}_{X,\theta}$ and the other terms have been identified before. 
Thus 
$ \gamma = (R^I)^{-1}U^{IO}(R^{O'})^{-1}D^{-1}_{\pil{}}$ and $\gamma$ is identified. 
\end{itemize}
\end{proof}

\section{Details of the Variational EM}
\label{appendix:vem}
The variational bound for the stochastic block model for multilevel network can be written as follows:
\begin{eqnarray*}
    \mathcal{I}_{\theta}(\mathcal{R}(\ZR{}, \ZL{} | \aff)) & =& 
     \sum_{j,l} \taul{jl}\log\pil{l}   + \sum_{i,k} \taur{ik} \sum_{j,l} \aff_{ij} \taul{jl}\log \gamma_{kl} \\
		&& + \quad \frac{1}{2} \sum_{i' \neq i} \sum_{k, k'} 
			\taur{ik}\taur{i'k'} \log \phi \left( \XR{ii'} , \alphar{kk'}\right) + \frac{1}{2} \sum_{j' \neq j} \sum_{l, l'} 
			\taul{jl}\taul{j'l'} \log \phi\left( \XL{jj'} \alphal{ll'} \right)\\
			&& - \sum_{i,k} \taur{ik}\log \taur{ik} - \sum_{j,l} \taul{jl}\log \taul{jl}
\end{eqnarray*}
The variational EM algorithm then consists on iterating the two following steps. At iteration $(t+1)$: 
\begin{description}
    \item[VE step] compute 
    \begin{align*} \{\taur{}, \taul{} \}^{(t+1)} & = \arg \underset{\taur{}, \taul{}}{\max}\; \mathcal{I}_{\theta^{(t)}}(\mathcal{R}(\ZR{}, \ZL{} | \aff))\\ 
    & = \arg \underset{\taur{}, \taul{}}{\min}\; \KL \left(\mathcal{R}(\ZR{}, \ZL{} | \aff) \| \mathbb{P}_{\theta^{(t)}}(\ZR{}, \ZL{} | \XR{}, \XL{}, \aff) \right) \,.
    \end{align*}
    \item[M step] compute 
    \[\theta^{(t+1)} = \arg \underset{\theta}{\max}\;  \mathcal{I}_{\theta}(\mathcal{R}^{(t+1)}(\ZR{}, \ZL{} | \aff)).\]
\end{description}
The variational parameters are sought by solving the equation: 
\begin{equation*}
    \Delta_{\taur{}, \taul{}}\left(\mathcal{I}_{\theta}(\mathcal{R}(\ZR{}, \ZL{}|\aff) + L(\taur{}, \taul{})\right)=0,
\end{equation*} where $L(\taur{}, \taul{})$ are the Lagrange multipliers for $\taur{i}$, $\taul{j}$ for all $i \in \{1, \dots, \nbr\}$, $j \in \{1, \dots, \nbr\}$.
There is no closed-form formula but when computing the derivatives, we obtain that the variational parameters follow the fixed point relationships:
\begin{align*}
	\widehat{\taul{jl}}  \propto & \pil{l} \prod_{i,k} \gamma_{kl}^{\aff_{ij}\widehat{\taur{ik}}}\prod_{j'\neq j}\prod_{l'}\phi(\XL{jj'}, \alphal{ll'})^{\widehat{\taul{j'l'}}} \\
    \widehat{\taur{ik}}  \propto &  \prod_{j,l} \gamma_{kl}^{\aff_{ij}\widehat{\taul{jl}}}\prod_{i'\neq i}\prod_{k'}\phi(\XR{ii'}, \alphar{kk'})^{\widehat{\taur{i'k'}}}\,,
\end{align*}
which are used in the VE step to update the $\taur{i}$'s  and $\taul{j}$'s.

On each update, the variational parameters of a certain level depend on both the parameter $\gamma$ and the variational parameters of the other level, which emphasizes the dependency structure of this multilevel model and the role of $\gamma$ as the dependency parameter of the model.
Notice also that when $\gamma_{kl} = \gamma_{kl'} = \pir{k}$ for all $l, l'$, that is the case of independence between the two levels then we can rewrite the fixed point relationships as follows:
\begin{align*}
	\widehat{\taul{jl}}  \propto & \pil{l}\prod_{j'\neq j}\prod_{l'}\phi(\XL{jj'}, \alphal{ll'})^{\widehat{\taul{j'l'}}} \quad \mbox{ and  } \quad 
    \widehat{\taur{ik}}  \propto   \pir{k}\prod_{i'\neq i}\prod_{k'}\phi(\XR{ii'}, \alphar{kk'})^{\widehat{\taur{i'k'}}},
\end{align*}
which is exactly the expression of the fixed point relationship of two independent SBMs.
Then, for the M step, we derive the following closed-form formulae:
\begin{align*}
	\widehat{\pil{l}} & =  \frac{1}{\nbl} \sum_j \widehat{\taul{jl}} & 
	\widehat{\alphal{ll'}} & = \frac{ \sum_{j' \neq j} \widehat{\taul{jl}}\XL{jj'}\widehat{\taul{j'l'}}}{\sum_{j' \neq j} \widehat{\taur{jl}}  \widehat{\taur{j'l'}} } \\
    \widehat{\gamma}_{kl} & = \frac{ \sum_{i,j}  \widehat{\taur{ik}} A_{ij} \widehat{\taul{jl}} }{ \sum_{i, j} A_{ij}\widehat{\taul{jl}} } &
	\widehat{\alphar{kk'}} & = \frac{ \sum_{i' \neq i} \widehat{\taur{ik}} \XR{ii'} \widehat{\taur{i'k'}}}{\sum_{i' \neq i} \widehat{\taur{ik}} \widehat{\taur{i'k'}} } 
\end{align*}
for which the gradient 
\begin{equation*}
    \Delta_{\theta}\left(\mathcal{I}_{\theta}(\mathcal{R}(\ZR{}, \ZL{}|\aff)) + L(\pil{}, \gamma )\right),
\end{equation*} is null.
The term $L(\pil{}, \gamma)$ contains the Lagrange multipliers for $\pil{}$ and $\gamma_{k\cdot}$ for all $k \in \{1, \dots, \QR\}$.

Model parameters have natural interpretations. $\pil{l}$ is the mean of the posterior probabilities for the organizations to belong to block $l$. 
$\alphar{kk'}$ (resp. $\alphal{ll'}$) is the ratio of existing links over possible links between blocks $k$ and $k'$ (resp. $l$ and $l'$). $\gamma_{kl}$ is the ratio of the number of individuals in block $k$ that are affiliated to any organization of block $l$ on the number of individuals that are affiliated to any organization of block $l$. If $\gamma$ is such that  the levels are independent, then any column of $\gamma$ represents the proportion of individuals in the different blocks: 
\begin{equation*}
    \pir{k} = \gamma_{k1} = \frac{1}{\nbr}\sum_{i}\widehat{\taur{ik}}.
\end{equation*}

\section{Details of the ICL criterion}
\label{appendix:icl}
We now derive an expression for the Integrated Complete Likelihood (ICL) model selection criterion.  Following \cite{daudin2008mixture}, the ICL is based on the  integrated complete likelihood i.e. the likelihood of the observations and the latent variables where the parameters have been  integrating  out against a prior distribution.    The latent variables $(\ZR{},\ZL{})$ being unobserved, they are imputed using the maximum a posteriori (MAP)  or $\hat \tau$.  We denote by $\widehat{\ZL{}}$ and $\widehat{\ZR{}}$ the inputed latent variables.  After imputation of the latent variables, an asymptotic approximation of this quantity   leads to the ICL criterion given in the paper (Equation \eqref{eq:icl}) and recalled here:
\begin{align*}
ICL(\QR, \QL) &= \log \ell_{\widehat{\theta}}(\XR{}, \XL{}, \widehat{\ZR{}}, \widehat{\ZL{}} | A, \QR, \QL) \\
    & - \frac{1}{2}\frac{\QR(\QR+1)}{2} \log \frac{\nbr(\nbr-1)}{2} - \frac{\QL(\QR-1)}{2} \log \nbr \\
  & - \frac{1}{2}\frac{\QL(\QL+1)}{2} \log \frac{\nbl(\nbl-1)}{2} - \frac{\QL-1}{2} \log \nbl . \\
\end{align*}
Let $\Theta = \Pi^O \times \mathcal{A}^I \times \mathcal{A}^O \times \Gamma$ be the space of the model parameters.  
We set a prior distribution on $\theta$: 
$$p(\theta| \QR,\QL) =  p(\gamma{}|\QR,\QL)p(\pil{}|\QL)p(\alphar{}|\QR)p(\alphal{}|\QL)$$
where   $p(\pil{}|\QL)$ is a Dirichlet distribution of hyper-parameter $(1/2, \cdots,1/2)$ and $p(\alphar{}|\QR)$ and $p(\alphal{}|\QL)$ are independent Beta distributions. \\
The marginal complete likelihood is written as follows:  
\begin{eqnarray}
\log \ell_{\theta} (\mathbf{X}, \mathbf{Z} | \aff, \QR, \QL) &=& 
	\log \left(\int_{\Theta} \ell_{\theta}( \XR{}, \XL{}, \ZR{}, \ZL{} | \theta, \aff, \QR, \QL) p(\theta| \QR, \QL) \text{d}\theta \right)  \nonumber\\
    & =&  \log \ell_{\alphar{}}(\XR{} | \ZR{}, \QR)  \label{marg:term1}   \\
&& +   \log  \ell_{\gamma}(\ZR{} | \aff, \ZL{}, \QR, \QL)       \label{marg:term2} \\
    &&      + \log \ell_{\alphal{},\pil{}}(\XL{} ,  \ZL{} | \QL)   \label{marg:term3} \,.
\end{eqnarray}
The quantity defined in \eqref{marg:term3} evaluated at $\ZL{} := \widehat{\ZL{}}$  is approximated as in \cite{daudin2008mixture} by 
\begin{equation}\label{eq:icl_term3}
\begin{array}{ccl}
 \log \ell_{\alphal{}}(\XL{} ,   \widehat{\ZL{}}, \QL)   &\approx&_{\nbl \rightarrow \infty}  \log \ell_{\widehat{\alphal{}},\widehat{\pil{}}}(\XL{}, \widehat{\ZL{}}  | \QL)  -    \mbox{pen}(\pil{},\alphal{}, \QL)\\
  \quad \mbox{pen}(\pil{},\alphal{}, \QL)  &=&  \frac{\QL - 1}{2} \log \nbl +\frac{1}{2}\frac{\QR(\QR+1)}{2} \log \frac{\nbr(\nbr-1)}{2}
\end{array}\,.
\end{equation}
This approximation results from a BIC-type approximation of $  \log \ell_{\widehat{\alphal{}}}(\XL{}| \widehat{\ZL{}},\QL) $ and a Stirling approximation of   $  \log \ell_{\pil{}}( \widehat{\ZL{}},\QL) $. \\ The same BIC-type approximation on  $ \log \ell_{\alphar{}}(\XR{} |  \widehat{\ZR{}} , \QR) $ (Equation \eqref{marg:term1})  leads to: 
\begin{equation} \label{eq:icl_term1}
\begin{array}{rcl}
 \log \ell_{\alphar{}}(\XR{} | \widehat{\ZR{}}, \QR)  &=&_{\nbr \rightarrow \infty}  \log \ell_{\widehat{\alphar{}}}(\XR{} |   \widehat{\ZR{}}, \QR) + \mbox{pen}(\alphar{}, \QR) \\
 \mbox{ with }      \mbox{pen}(\alphar{}, \QR) & =& \frac{1}{2}\frac{\QR(\QR+1)}{2} \log \frac{\nbr(\nbr-1)}{2} \\
 \end{array}\,.
  \end{equation}

For quantity \eqref{marg:term2} depending on $\gamma$ and $\ZR{}$ given $(\QR, \QL)$, we have to adapt the calculus. Let us set  independent  Dirichlet prior distributions of order $\QR$   $\mathcal{D}(1/2, \dots, 1/2)$ on the columns $\gamma_{\cdot l}$. We are able to derive an exact expression of $   \log  \ell_{\gamma}(\ZR{} | \aff, \ZL{}, \QR, \QL) $:  
\begin{align}
	\ell_{\gamma}(\ZR{} | \aff, \ZL{}, \QR, \QL) 
	& = \int \ell (\ZR{} | \aff, \ZL{}, \gamma, \QR, \QL) p(\gamma, \QR, \QL) \text{d}\gamma \nonumber\\
    & = \prod_{i}\int \prod_{j, k, l} \gamma_{kl}^{\aff_{ij}\ZR{ik}\ZL{jl}}p(\gamma_{kl})\text{d}\gamma_{kl} \nonumber\\
    & = \prod_{l} \int \prod_{k} \gamma_{kl}^{N_{kl}}p(\gamma_{kl})\text{d}\gamma_{kl}, \quad \mbox{where} \quad  N_{kl} = \sum_{ij}\aff_{ij}\ZR{ik}\ZL{jl}\nonumber \\
    &=  \prod_{l} \int \prod_{k} \gamma_{k,l}^{N_{kl} + a - 1}\frac{\Gamma(1/2\cdot \QR)}{\Gamma(1/2)^{\QR}}\text{d}\gamma_{kl} \nonumber\\ 
    &= \frac{\Gamma(1/2\QR)^{\QL}}{\Gamma(1/2)^{\QL+\QR}} \prod_{l} \frac{\prod_{k}\Gamma(N_{kl} + 1/2)}{\Gamma(1/2\QR + \sum_{k}N_{kl})} \,.\nonumber 
\end{align}
Now,  using the fact that  $\log \Gamma(n+1) \overset{n \to \infty}{\sim} (n + 1/2) \log n + n$, we obtain: 
\begin{equation}\label{eq:ICL17}
\begin{array}{ccl}
\log \ell_{\gamma}(\ZR{} |\aff, \ZL{},\QR, \QL) &\approx&_{ (\nbl,\nbr) \rightarrow \infty}  \sum_{k,l} (N_{kl}  \log N_{kl} + N_{kl}) \\  
&& - \sum_{l} 
 \left(\frac{\QR - 1}{2} + \sum_{k}N_{kl}\right)\log\left(\sum_k N_{kl} \right) - \sum_{k,l} N_{kl}.
\end{array}
\end{equation}
The quantity  \eqref{eq:ICL17} evaluated at $(\ZR{}, \ZL{} ):= (\widehat{\ZR{}}, \widehat{\ZL{}})$  can be reformulated in the following way: 
\begin{equation*}
\begin{array}{ccl}
\log \ell_{\gamma}(\widehat{\ZR{}} |\aff, \widehat{\ZL{}},\QR, \QL)   &\approx&_{ (\nbl,\nbr) \rightarrow \infty} \log \ell_{\hat{\gamma}} ( \widehat{\ZR{}}  | \aff,  \widehat{\ZL{}},\QR, \QL ) - \frac{\QR - 1}{2} \sum_{l} \log \sum_{i,j} \aff_{ij}\widehat{\ZL{jl}}\\
\mbox{ with }\hat{\gamma}_{kl}  &=&  \frac{ \sum_{i,j} \widehat{\ZR{ik}}\aff_{ij}\widehat{\ZL{jl}} }{ \sum_{i,j} \aff_{ij}\widehat{\ZL{jl}} }
\end{array}
\end{equation*}
Noticing that 
$
    \log \sum_{i,j} \aff_{ij}\widehat{\ZL{jl}} = \log \nbr + \log \frac{\sum_{i,j} \aff_{ij}\widehat{\ZL{jl}}}{\nbr}   = O(\log  \nbr)
$
leads to 
\begin{equation}\label{eq:icl_term2}
\log \ell_{\gamma}(\widehat{\ZR{}} |\aff, \widehat{\ZL{}},\QR, \QL)   \approx_{ (\nbl,\nbr) \rightarrow \infty} \log \ell_{\hat{\gamma}} ( \widehat{\ZR{}}  | \aff,  \widehat{\ZL{}},\QR, \QL ) - \frac{\QR - 1}{2} \QL \log  \nbr\,.
\end{equation}
Combining Equations \eqref{eq:icl_term3}, \eqref{eq:icl_term1} and \eqref{eq:icl_term2} we obtain the given expression.

\end{document}